\documentclass[12pt, onecolumn]{IEEEtran}
\usepackage{url}
\usepackage{amssymb,amsmath}
\usepackage[all]{xy}%\CompileMatrices
\usepackage{bm}
\usepackage{subfigure}
\usepackage{xspace}
\usepackage{graphicx,color,psfrag}
\usepackage{soul}
\usepackage{cite,xspace}
\usepackage[lined,boxed,commentsnumbered]{algorithm2e}
\newcommand{\mystretch}{1.75}
\renewcommand{\baselinestretch}{\mystretch}

\begin{document}

\title{Optimal Energy Allocation \\ for Wireless Communications \\ with Energy Harvesting Constraints}%

%\author{\authorblockN{Chin Keong Ho$^*$}
%\authorblockA{$^*$Institute for Infocomm Research, A*STAR,% 1 Fusionopolis Way, %\#21-01 Connexis,
%%Singapore 138632.
%\\e-mail: \{hock,rzhang\}@i2r.a-star.edu.sg}
%\and
%\authorblockN{Rui Zhang$^*$$^\dag$}
%\authorblockA{$^\dag$Department of ECE, National University of Singapore\\e-mail: elezhang@nus.edu.sg}}

\author{Chin Keong Ho,~\IEEEmembership{IEEE Member}, and~Rui Zhang,~\IEEEmembership{IEEE Member}%
\thanks{This paper was presented in part at the IEEE International Symposium on Information Theory, Austin, TX, June 2010.}
\thanks{C. K. Ho is with the Institute for Infocomm Research, A*STAR, 1 Fusionopolis Way, \#21-01 Connexis, Singapore 138632 (e-mail: hock@i2r.a-star.edu.sg).}
\thanks{R. Zhang is with the Department of Electrical and Computer Engineering, National University of Singapore (e-mail:elezhang@nus.edu.sg). He is also
with the Institute for Infocomm Research, A*STAR, Singapore. This work has been supported in part by the National University of Singapore under Research Grant R-263-000-679-133.}
}

\newtheorem{conjecture}{Conjecture}
\newtheorem{remark}{Remark}
\newtheorem{insight}{Insight}
\newtheorem{question}{Question}
\newtheorem{proposition}{Proposition}
\newtheorem{corollary}{Corollary}
\newtheorem{lemma}{Lemma}
\newtheorem{assumption}{Assumption}
\newtheorem{theorem}{Theorem}
\newtheorem{example}{Example}
\newtheorem{property}[theorem]{Property}

\newcommand{\myse}{\IEEEyessubnumber} % setting subnumbering
\newcommand{\myses}{\myse\IEEEeqnarraynumspace} % setting subnumbering, shifted

\newcommand{\set}[1]{\mathcal{#1}}

% standard enumerate
\newcommand{\bn}{\begin{enumerate}}
\newcommand{\en}{\end{enumerate}}

% standard itemize
\newcommand{\bi}{\begin{itemize}}
\newcommand{\ei}{\end{itemize}}

% standard equation array
\newcommand{\be}{\begin{IEEEeqnarray}{rCl}}
\newcommand{\ee}{\end{IEEEeqnarray}}

% no equation numbering and no labeling
\newcommand{\benl}{\begin{IEEEeqnarray*}}
\newcommand{\eenl}{\end{IEEEeqnarray*}}

% no labeling
\newcommand{\bel}{\begin{IEEEeqnarray}}
\newcommand{\eel}{\end{IEEEeqnarray}}

% no equation numbering
\newcommand{\ben}{\begin{IEEEeqnarray*}{rCl}}
\newcommand{\een}{\end{IEEEeqnarray*}}

\newcommand{\barr}{\begin{array}}
\newcommand{\earr}{\end{array}}

\newenvironment{definition}[1][Definition:]{\begin{trivlist}
\item[\hskip \labelsep {\it #1}]}{\end{trivlist}}

\newcommand{\ud}{\mathrm{d}} % d in integration

\newcommand{\FigSize}{1}
\newcommand{\FigSizeSmall}{0.5}

\newcommand{\avesnr} {\bar{\gamma}} % average SNR
\newcommand{\snr} {\gamma} % instantaneous SNR

\newcommand{\re}[1]{(\ref{#1})}

\newcommand{\Pe} {P_{\mathrm {e}}} % BER

\newcommand{\goodgap}{%
\hspace{\subfigtopskip}%
\hspace{\subfigbottomskip}}

\newcommand{\dhat}[1]{\Hat{\Hat{#1}}} % double hat
\newcommand{\that}[1]{\Hat{\Hat{\Hat{#1}}}} % triple hat
\newcommand{\dtilde}[1]{\Tilde{\Tilde{#1}}} % double tilde
\newcommand{\ttilde}[1]{\Tilde{\Tilde{\Tilde{#1}}}} % triple hat

\newcommand{\trace}[1]{\mathrm{tr}\{#1\}}

\newcommand{\mi}{I}
\newcommand{\reward}{\mathsf{r}}
\newcommand{\TPaveopt}{\TPave^{\star}}
\newcommand{\TPave}{\mathsf{T}}
\newcommand{\valuefun}{h}
\newcommand{\statenew}{\mathsf{s}'}
\newcommand{\state}{\mathsf{s}}
\newcommand{\statespace}{\mathcal{S}}
\newcommand{\chanspace}{\Gamma}
\newcommand{\Hspace}{\mathcal{H}}
\newcommand{\battspace}{\mathcal{B}}
\newcommand{\CuBrac}[1]{\left\{#1\right\}}
\newcommand{\action}{\mathsf{a}}
\newcommand{\setT}{\mathcal{S}}

\renewcommand{\FigSize}{1}%{0.8}

\maketitle

%\vspace{-2.25cm}

\begin{abstract}
We consider the use of energy harvesters, in place of conventional batteries with fixed energy storage, for point-to-point wireless communications. In addition to the challenge of transmitting in a channel with time selective fading, energy harvesters provide a perpetual but unreliable energy source. In this paper, we consider the problem of energy allocation over a finite horizon, taking into account channel conditions and energy sources that are time varying, so as to maximize the throughput. Two types of side information (SI) on the channel conditions and harvested energy are assumed to be available: causal SI (of the past and present slots) or full SI (of the past, present and future slots). We obtain structural results for the optimal energy allocation, via the use of dynamic programming and convex optimization techniques. In particular, if unlimited energy can be stored in the battery with harvested energy and the full SI is available, we prove the optimality of a water-filling energy allocation solution where the so-called water levels follow a staircase function.
\end{abstract}

%\newpage

\section{Introduction}

In conventional  wireless communication systems, the communication devices have access to a fixed power supply, or are powered by replaceable or rechargeable batteries to enable user mobility. In these cases, the transmissions are limited by power constraints for safety reasons, or by the sum energy constraint so as to prolong operating time for battery-powered devices.
In other communication systems, however, a fixed power supply is not readily available, and even replacing the batteries periodically may not be a viable option if the replacement is considered to be too inconvenient (when thousands of senor nodes are scattered throughout the building), too dangerous (the devices may be located in toxic environments) or even impossible (when the devices are embedded in building structures or inside human bodies).
In such situations, the use of energy harvesting for wireless communications appears appealing or sometimes even essential.
Examples of energy that can be harvested include solar energy, piezoelectric energy and thermal energy, etc.

For transmitters that are powered by energy harvesters, the energy that can potentially be harvested is unlimited. Typically, energy is replenished by the energy harvester, while expended for communications or other processing; any unused energy is then stored in an energy storage,  such as a rechargeable battery.
However, unlike conventional communication devices that are subject only to a power constraint or a sum energy constraint, transmitters with energy harvesting capabilities are, in addition, subject to other {\em energy harvesting constraints}.
Specifically, in every time slot, each transmitter is constrained to use at most the amount of stored energy currently available, although more energy may become available in the future slots.
Thus, a causality constraint is imposed on the use of the harvested energy.

Several contributions in the literature have considered using energy harvester as an energy source, in particular based on the technique of dynamic programming \cite{Bertsekas}. In \cite{Fu03}, the problem of maximizing a reward that is linear with the energy used is studied.
In  \cite{Sharma10}, the discounted throughput is maximized over an infinite horizon, where queuing for data is also considered.
In \cite{Kansal}, adaptive duty cycling is employed for throughput maximization and implemented in practical systems.
% In \cite{YangUlukus10}, the minimum time required to complete the transmission of a fixed amount of data is studied.
In \cite{OzelUlukus10}, an information-theoretic approach is considered  where the energy is harvested at the level of channel uses.
In \cite{Ozeletc11}, and the references therein, optimal approaches are also considered for throughput maximization over AWGN or fading channels.

In this work, we consider the problem of maximizing the throughput via energy allocation over a finite horizon of $K$ time slots. The channel signal-to-noise ratios (SNRs) and the amount of energy harvested change over different slots.
Our aim is to study the structure of the maximum throughput and the corresponding optimal energy allocation solution, such as concavity and monotonicity. These results may be useful for developing heuristic solutions, since the optimal solutions are often complex to obtain in practice.
We consider two types of side information (SI) available to the transmitter:
\bi
\item {\em causal SI}, consisting of past and present channel conditions, in terms of SNR, and the amount of energy harvested in the past slots, or
\item {\em full SI}, consisting of past, present and future channel conditions and amount of energy harvested.
\ei
The case of full SI may be justified if the environment is highly predictable, e.g., the energy is harvested from the vibration of motors that are turned on only during fixed operating hours and line-of-sight is available for communications.

Our contributions are as follows.
Given causal SI, and assuming that the variations in the channel conditions and energy harvested are modeled by a first-order Markov process, we obtain the optimal energy allocation solution by dynamic programming, which can be computed offline and stored in a lookup table for implementation. Moreover, we obtain structural results to characterize the optimal solution.
% added in revision
%The optimal energy allocation can be computed offline and stored in a lookup table for implementation.
%During online operation, the device simply inputs the current state (comprising the battery level and channel state information) into the table to obtain the required energy for transmission. The complexity of the online operation is thus very low.
%
% proposed simple heuristic algorithm that works well in practice
Given full SI, we obtain a closed-form solution for $K=2$ slots. We also obtain the structure of this optimal solution for arbitrary $K$ with unlimited energy storage. The optimal solution then has a water-filling interpretation, as in \cite{Goldsmith97IT}. However, instead of a single water level, there are multiple so-called water levels that are non-decreasing over time, i.e., the water levels follow a staircase-like function.
Finally, we propose a heuristic scheme that uses only causal SI. Compared to a naive scheme, the proposed scheme performs relatively close to the optimal throughput obtained with full SI in our numerical studies.

This paper is organized as follows.
Section~\ref{sec:intro} gives the system model.
Then, Section~\ref{sec:causal} considers optimal schemes with availability of causal SI of the channel conditions and harvested energy. Section~\ref{sec:full-arbBmax} considers optimal schemes with availability of full SI with a constraint on the maximum amount of energy that can be stored on the battery, while Section~\ref{sec:full-infBmax} considers the specific case where this constraint is removed.
%Section~\ref{sec:heu} draws insights from the last two sections to obtain heuristic schemes with availability of causal side information.
Section~\ref{sec:num} shows numerical results for the various schemes.
Finally, Section~\ref{sec:con} concludes the paper.

%We write $X_1, \cdots, X_K$ collectively as a length-$K$ vector $X^K$.
% We say the function $f(x)$ is increasing (decreasing) in $x$ if $f(x')\geq (\leq) f(x)$ for all $x'\geq x$.

\section{System Model}\label{sec:intro}
\newcommand{\Bmax}{B_{\mathrm{max}}}

\begin{figure}%[f]
\centering
%\psfrag{h}{$H_k$}
%\psfrag{t}{${T_k}$}
%\psfrag{w}{$W_k$}
%\psfrag{x}{$X^n_k$}
%\psfrag{B}{$B_{k}$}
%\psfrag{s}{\hspace{-0.15cm}$\sqrt{T_k}X^n_k$}
\includegraphics[scale=1.25]{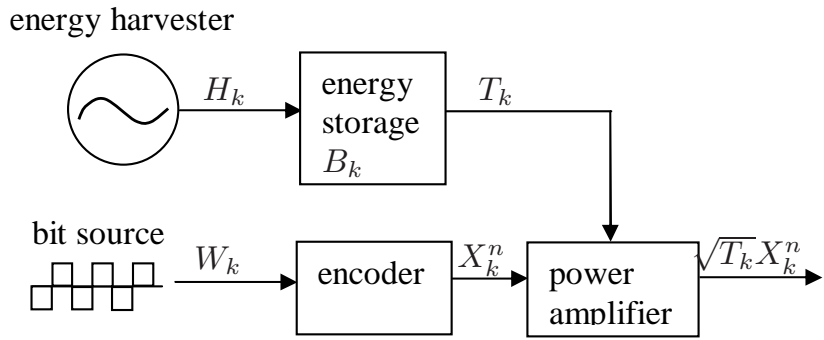}
\caption{Block diagram of a transmitter powered by an energy harvester. Energy is replenished by an energy harvester but is drawn for transmission.}
\label{fig:diag}
\end{figure}

For simplicity, each packet transmission is performed in one time slot. Each time slot allows $n$ symbols to be transmitted, where $n$ is assumed to be sufficiently large for reliable decoding. We index time by the slot index $k\in \mathcal{K}\triangleq \{1, \cdots, K\}$.
%For convenience of description, let $k^-$ be the time instant just before slot $k$. %We denote $t^{--}$ as the time instant just before $t^-$.
We assume that there is always back-logged data available for transmission.
In slot $k$, a message $W_k\in\{1,\cdots, 2^{nR_k}\}$ is sent, where the rate $R_k\geq 0$ in bits per symbol can be selected.

We consider a point-to-point, flat-fading, single-antenna communication system.
As shown in Fig.~\ref{fig:diag}, the transmitter is powered by an energy harvester that takes in harvested energy as input, that is then stored in an energy storage. The bits to be sent are encoded by an encoder, then sent by a power amplifier that uses the stored energy.
Energy is measured on a per symbol (or channel use) basis, hence we use the terms energy and power interchangeably.

Consider slot $k\in\mathcal{K}$.
At time instant $k^-$, which denotes the time instant just before slot $k$, the battery has available $B_{k}\geq 0$ amount of stored energy per symbol. For transmission, the message $W_k$ is first encoded as data symbols $X^n_k\triangleq [X_{1k}, \cdots, X_{nk}]$ of length $n$, where we normalize $\sum_{i=1}^n |X_{ik}|^2/n=1.$
Then the transmitter transmits packet $k$ in slot $k$ as $\sqrt{T_k}X^n_k$, where $0\leq T_k\leq B_{k}$ is the energy per symbol used by the power amplifier.
Except for transmission, we assume the other circuits in the transmitter consume negligible energy.

\subsubsection{Mutual Information}
We assume the channel is quasi-static for every slot $k\in\mathcal{K}$ with SNR $\snr_k$.
The maximum reliable transmission rate in slot $k$ is then given by the mutual information $\mi(\snr_k, T_k)\geq 0$ in bits per symbol \cite{Cover06}.
In general, we assume that $\mi(\snr, T)$ is concave in $T$ given $\snr$, and is increasing in $T$ for all $\snr$.
For example, we may employ Gaussian signalling for transmission over a complex Gaussian channel \cite{Cover06},  which gives
\be\label{eqn:gaussianmi}
\mi(\snr, T)=\log_2(1+T\snr).
\ee

%The maximum reliable transmission rate is given by the mutual information
%\be\label{eqn:reward}
%\mi(\snr_k, T_k)=\log(1+T_k \snr_k)
%\ee
%in bits per symbol, assuming complex-valued Gaussian signaling is used.

\subsubsection{Battery Dynamics}

In general, let us denote a vector  of length $k$ as $Y^{k}=[Y_1, \cdots, Y_k]$, e.g., the battery energy from slot $1$ to slot $k$ is given by $B^{k}=[B_1, \cdots, B_k]$.
While transmitting packet $k$, the energy harvester collects an average energy of $H_{k}\geq 0$ per symbol, which is then stored in the battery.
At time instant $(k+1)^-$, the energy stored is updated in general as
\ben%\label{eqn:battery0}
B_{k+1} = f\left(B^{k}, T^{k}, H^{k}\right), k\in\mathcal{K},
\een
where the function $f$ depends on the battery dynamics, such as the storage efficiency and memory effects.
Intuitively, we expect $B_{k+1}$ to increase (or remains the same) if $B_{k}$ or $H_k$ increases, or if $T_{k}$ decreases. As a good approximation in practice, we assume the stored energy increases and decreases linearly provided the maximum stored energy in the battery $\Bmax$ is not exceeded, i.e.,
\be\label{eqn:battery1}
B_{k+1} = \min \{B_{k} -  T_{k} + H_k, \Bmax\}, k\in\mathcal{K}.
\ee
We assume the initial stored energy $B_1$ is known, where  $0\leq B_1\leq \Bmax$.
Thus, $\{B_{k}\}$ follows a deterministic first-order Markov model that depends only on the immediate past random variables.

\subsubsection{Channel and Harvest Dynamics}
%Let $\state_k \triangleq (\snr_{k}, H_{k})$ for $k=0,\cdots,K$ denotes the present state of the system at time $k$. Also, let ${\state}^k\triangleq (\snr^k, H^k)$ for $k=1,\cdots,K$  denotes the present and past states of the system at time $k$.

To model the unpredictable nature of energy harvesting and the wireless channel over time, we model%
 \footnote{The harvested energy in slot $K$, namely $H_K$, cannot be used for transmission in slots $1$ to slot $K$ and so does not affect the throughput. }
 $H^{K-1}$ and $\snr^K$ jointly as a random process described by their joint distribution.  The exact distribution  depends on the energy harvester used and the wireless channel environment.

In most typical operating scenarios, both the wireless channels and the harvested energy vary slowly over time.
To account for these variations, the SNR $\snr_k$ is assumed to be constant in each slot and follow a first-order stationary Markov model over time $k$, see e.g. similar assumptions in \cite{SadeghiKennedyRapajicShams08}. Also, the harvested energy $H_k$  is modeled as  first-order stationary Markov model over time $k$, where the accuracy of this model is justified by empirical studies when solar energy is harvested \cite{hoICCS10}.
%To yield tractable analysis, we model the variation over time with a first-order stationary Markov model.
Given $H_0=\hat{H}_0$ and $\snr_1=\hat{\snr}_1$, the joint pdf of  $H^{K-1}$ and $\snr^K$ thus becomes
%
%In wireless channels, the channel between the transmitter and receiver varies over time.
%To model the channel variations, let $\snr_k$ be the SNR at the receiver when packet $k$ is transmitted, described by the pdf $p_{\snr^K}(\snr^K)$.
%We normalize $\snr^K$ such that $\mathbb{E}[\snr_k]=1$ for all $k$ and assume $\snr^K$ and $H^K$ are independent. For Rayleigh-fading channel,  $\snr_k$ is chi-squared distributed with two degrees of freedom. The SNR of the channel when packet $k$ is transmitted is given by $T_k \snr_k/T$, where $T$ is the transmission time. Without loss of generality, we let $T=1$.
%
%For simplicity, we assume i.i.d harvest power and SNR with joint distribution
\be\nonumber
&& p_{H^{K-1}\hspace{-0.00cm}, \snr^K}(H^{K-1}\hspace{-0.cm}, \snr^K|H_0=\hat{H}_0, \snr_1=\hat{\snr}_1)
\\
%p_{{\state}^K}\left({\state}^K\right)=
=&& \prod_{k=3}^K \hspace{-0.00cm} p_{H_{k-1}}(H_{k-1}|H_{k-2}) p_{\snr_k}(\snr_k|\snr_{k-1}) \nonumber \\
&& \times p_{H_{1}}(H_{1}|H_0=\hat{H}_0) p_{\snr_2}(\snr_2|\snr_1=\hat{\snr}_1)
%&& \cdot \delta(\snr_1-\hat{\snr}_1) \cdot \delta({H}_0-\hat{H}_0)
\label{eqn:distHsnr}
%p_{H^{K-1}\hspace{-0.00cm}, \snr^K}(H^{K-1}\hspace{-0.1cm}, \snr^K|H_0 , \snr_1)\hspace{-0.1cm}=\hspace{-0.175cm}
%%p_{{\state}^K}\left({\state}^K\right)=
%\prod_{k=2}^K \hspace{-0.05cm} p_{H_{k-1}}(H_{k-1}|H_{k-2}) p_{\snr_k}(\snr_k|\snr_{k-1})\;\;\;\;\;
\ee
where $p_{H_k}(\cdot|\cdot)$ and $p_{\snr_k}(\cdot|\cdot)$ are independent of $k$.
In \re{eqn:distHsnr}, we have also assumed that the harvested energy and the SNR are independent, which is reasonable in most practical scenarios.
%We assume the channel power and the harvest energy to be independent. We assume both $H_k$ and $\snr_k$ to be i.i.d., i.e., $p_{H^K}(H^K)=\Pi_k p_H(H_k)$ and $p_{\snr^K}(\snr^K)=\Pi_k p_{\snr}(\snr_k)$. It is straightforward to extend the results to the case when either one or both $H_k, \snr_k$ are Markovian, i.e., $p_{H^K}(H^K)=\Pi_k p_H(H_k|H_{k-1})$ and $p_{\snr^K}(\snr^K)=\Pi_k p_{\snr}(\snr_k|\snr_{k-1})$.
In this paper, we assume that the joint distribution \re{eqn:distHsnr} is known, which may be obtained via long-term measurements in practice.

\subsubsection{Overall Dynamics}

Let us denote the {\em state} $\state_k= (\snr_{k}, H_{k-1}, B_{k}), k\in\mathcal{K}$, or simply $\state$ if the index $k$ is arbitrary. Let the accumulated states be $\state^{k} \triangleq (\state_1, \cdots, \state_{k}),$ $k\in\mathcal{K}$.

We assume the initial state $\state_1\triangleq (\snr_{1}, H_{0}, B_1)$ to be always known at the transmitter, which may be obtained causally prior to any transmission.
From \re{eqn:battery1} and \re{eqn:distHsnr}, given $\state_1=\hat{\state}_1$, the states thus follow a first-order Markov model:
\be\label{eqn:distall}
p_{\state^{K}}(\state^{K}|{\state}_1=\hat{\state}_1)=
\prod_{k=3}^K p_{\state_{k}}(\state_{k}|\state_{k-1}) \times  p_{\state_{2}}(\state_{2}|{\state}_1=\hat{\state}_1). %\cdot \delta(\state_1-\hat{\state}_1).
\ee
In particular, \re{eqn:distall} includes the special cases where the states are independent, i.e., $p_{\state_{k}}(\state_{k}|\state_{k-1})=p_{\state_{k}}(\state_{k})$, or where the states are deterministic rather than random, i.e., $p_{\state_{k}}(\state_{k}|\state_{k-1})=\delta(\state_{k}-\hat{\state}_{k})$, where $\delta(\cdot)$ is the Dirac delta function.

In the next three sections, we consider the problem of maximizing the throughput subject to energy harvesting constraints, given either causal SI or full SI.

\section{Causal Side Information}\label{sec:causal}
\subsection{Problem Statement}

%We consider throughput maximization over a finite horizon. %The problem can also be extended to the case of infinite horizon.

We first consider the case of causal SI, in which the transmitter is given knowledge%
\footnote{It can be shown that having knowledge of previous states $\state^{k-1}$ does not improve throughput, due to the Markovian property of the states in \re{eqn:distall}.}
of $\state_{k}$ before packet $k$ is transmitted, where $k\in\mathcal{K}$.
% >> added in revision
That is, at slot $k$ the transmitter only knows the present channel SNR $\snr_{k}$, past harvested energy $H_{k}$ and present energy stored in the battery $B_k$.
In practice, for instance, the receiver feeds back $\snr_{k}$ shortly before transmission, while the transmitter infers $H_{k-1}$ and $B_k$ from its energy storage device.
%For example, the transmitter only knows ${\state}_0$ before transmitting packet $1$.
%After transmission, state $\state_1$ is additionally made available to the transmitter. Hence, the transmitter knows $(\state_0, \state_1)$ before the second packet is transmitted, and so on.
We say that causal SI is available as future states are not {\it a priori} known.
Thus, this allows us to model and treat the unpredictable nature of the wireless channel and harvesting environment.

The causal SI is used to decide the amount of energy $T_k$ for transmitting packet $k$.
We want to maximize the throughput, i.e., the expected mutual information summed over a finite horizon of $K$ time slots, by choosing a deterministic power allocation policy $\pi=\{T_k({\state}_{k}), \forall \state_{k}, k=1,\cdots, K\}$.
The policy can be optimized offline and implemented in real time via a lookup table that is stored at the transmitter.

A policy is feasible if the energy harvesting constraints $0\leq T_k(\state_{k})\leq B_k$ is satisfied for all possible $\state^k$ and all $k \in\mathcal{K}$; we denote the  space of all feasible policies as $\Pi$.
Mathematically, given $\state_1$, the maximum throughput is
\be\label{eqn:prob}
\mathcal{T}^{\star} &=& \max_{\pi\in\Pi} \mathcal{T}(\pi),
\ee
where
\be\label{eqn:T}
\mathcal{T}(\pi)
&=& \sum_{k=1}^K \mathbb{E}\left[\mi(\snr_k, T_k({\state}_{k}))|{\state}_1, \pi \right].
\ee
In \re{eqn:T}, the $k$th summation term represents the throughput of packet $k$ (after expectation); its expectation is performed over all (relevant) random variables given initial state $\state_1$ and policy $\pi$.

For example, if $K=2$ and a given policy, \re{eqn:T} simplifies as
\be
%&&  \\
%&=&
%\mathbb{E}_{{\snr}_1}\left[\mi(\snr_1, T_1({\state}_0))|{\state}_0, \pi \right]
%+
%\mathbb{E}_{{\state}_1}\left[\mi(\snr_2, T_2({\state}^1))\left|{\state}_0, \pi \right]\right.
%%.
%\nonumber \\
\mathcal{T} &=
\mi(\snr_1, T_1({\state}_1)) %
+
\mathbb{E}_{{\state}_2}\left[\mi(\snr_2, T_2({\state}_2))  \Big|
{\state}_1 \right]
%\nonumber
%\\ \nonumber
%&=&
%\mathbb{E}_{{\state}_1}\left[\mi(\snr_1, T_1({\state}_0))
%+
%\mathbb{E}_{{\state}_2}\left[\mi(\snr_2, T_2({\state}^1))|{\state}_1, \pi \right]
%\Big|{\state}_0, \pi
%\right].
%&=
%\mathbb{E}_{{\snr}_1}\Big[&\mi(\snr_1, T_1({\state}_0))+
%\mathbb{E}_{{\snr}_2}\left[\mi(\snr_2, T_2({\state}_1))|{\snr}_1, \pi \right]
%\Big|{\snr}_0, \pi
%\Big]
%\;\;\;
%&=&\int_{{\snr}_1}  \mi(\snr_1, T_1({\state}_0)) p(\snr_1|\snr_0) \ud \snr_1
%+
%\int_{{\snr}_2}  \mi(\snr_2, T_2({\state}^1)) p(\snr_1|\snr_0) \ud \snr_1
\label{eqn:ex1}\ee
subject to $0\leq T_1\leq B_1$ for the first term and $0\leq T_2\leq B_2=\min \{B_{1} -  T_{1} + H_1, \Bmax\}$ for the second term.
%The inner expectation is expressed as $\mathbb{E}_{{\snr}_2}\left[\mi(\snr_2, T_2({\state}^1))|{\snr}_1,  \pi \right]$
%subject to  $0\leq T_2\leq B_2$.
%The second equality is due to  \re{eqn:distHsnr} and \re{eqn:distall}.
%The policy $\pi$ is feasible if
Clearly, the transmission energy $T_1$ in the first slot affects the stored energy $B_2$ available in the second slot, which in turn affects the energy $T_2$ to be allocated.

In general the optimization of $\{T_k\}$ cannot be performed independently due to the energy harvesting constraints, as shown also in the above example. %In the above example, $T_2$ is constrained by $B_2$, which in turn depends on $T_1$.
Instead, for the above example, we can first optimize $T_2$ given all possible $T_1$ (and hence all possible $B_2$), then optimize for $T_1$ with $T_2$ replaced by the optimized value (as a function of $T_1$).
This approach, as will be suggested by dynamic programming in the general case, will be shown to be optimal.

%The computation of each expected mutual information term in \re{eqn:T} depends only on $\snr_0$ and $\pi$, assuming the joint pdf factors as \re{eqn:distHsnr}; this is evident from \re{eqn:ex1} for $K=2$.
%That is, the harvest energy $\{H_k\}$ and the stored battery $\{B_k\}$ does not affect $\mathcal{T}(\pi)$. However, to check if a policy is feasible requires the knowledge of $B^K.$
%Hence, in general a sufficient state for obtaining $\mathcal{T}^{\star}$ is the state $\tilde{\state}^{k-1} \triangleq (\tilde{\state}_0, \cdots, \tilde{\state}_{k-1})$
%
%$\tilde{\state}_k\triangleq (\snr_{k-1}, B_k)$ and let the accumulated state  for $k\in\mathcal{K}$.

%\be\label{eqn:T1}
%\mathcal{T}(\pi)
%&=& \sum_{k=1}^K \mathbb{E}_{{\state}_k}\left[\mi(\snr_k, T_k({\state}^k))|{\state}_0, \pi \right]
%\ee

\subsection{Optimal Solution}

%state $\state_k$ is available prior to every transmission.

The optimization problem \re{eqn:prob} is solved by dynamic programming in Lemma~\ref{lem:1}.

%We denote the SNR in the next slot as $\tilde{\snr}$. %Thus, given SNR in the present slot as $\snr$, the

\begin{lemma}\label{lem:1}
%Given $\state_1=(\snr_1, B_1)$, the maximum throughput is given by
%\be\label{eqn:opt}
%J_1(\snr_1, B_1)= \max_{\pi\in\Pi} \sum_{k=1}^K \mathbb{E}\left[\mi(\snr_k, T_k)|\snr_k, T_k, \pi\right].
%\ee
%where the expectation is performed over all random variables $H^K, \snr^K$ given $\snr_k, B_k$ using policy $\pi$.
Given initial state $\state_1=(\snr_1,H_0,B_1)$, the maximum throughput
$\mathcal{T}^{\star}$ is given by $J_1(\state_1)$, which can be computed recursively based on Bellman's equations, starting from $J_K(\state_K)$,  $J_{K-1}(\state_{K-1})$, and so on until $J_{1}(\state_1)$:
\be\label{eqn:J}
J_{K}(\snr,H,B) &=& \max_{0\leq T\leq B}
{\mi}(\snr, T)
= {\mi}(\snr, B) ,
%\label{eqn:bellman}
%\log(1+\snr T) = \log(1+\snr B),
\IEEEyessubnumber
\label{eqn:J_K1}\\
J_k(\snr,H,B) &=& \max_{0\leq T\leq B} {\mi}(\snr, T)
%\log(1+\snr T)
+ \bar{J}_{k+1}(\snr, H, B-T)
\;\;\;\;\;\;
\IEEEyessubnumber
\label{eqn:J_k}
\ee
for $k=1,\cdots,K-1$,
where
\be\nonumber %\label{eqn:I_expected}
%\bar{\mi}(\snr, T) &=& \mathbb{E}_{\tilde{\snr}}[ \mi(\tilde{\snr}, T)|\snr]  \\
 && \bar{J}_{k+1}(\snr,H,x)  \\
 &=& \mathbb{E}_{\tilde{H}, \tilde{\snr}}\left[ J_{k+1}(\tilde{\snr}, \tilde{H}, \min\{\Bmax, x+\tilde{H}\}) \big|\snr, H \right].
\label{eqn:J_expected}
\ee
In \re{eqn:J_expected}, $\tilde{H}$ denotes the harvested energy in the present slot given the harvested energy ${H}$ in the past slot, and
 $\tilde{\snr}$ denotes the SNR  in the next slot given the SNR ${\snr}$ in the present slot.
%$H_{k}$ and $\snr_{k+1}$.
An optimal policy is denoted as $\pi^{\star}=\{T^{\star}_k({\state}_{k}), \forall \state_{k}, k=1,\cdots, K\}$, where $T^{\star}_k(\state_{k})$ is the optimal $T$ that solves \re{eqn:J}.
\end{lemma}
\begin{proof}
The proof follows by applying Bellman's equation  \cite{Bertsekas} and using \re{eqn:battery1} and \re{eqn:distHsnr}.
%It is easy to see this is true if $K=2$. Given $\state_1$, the maximum second term in \re{eqn:ex1} is maximized by
\end{proof}

In \re{eqn:J_K1}, the optimal maximization  is trivial: the interpretation is that we use all available energy for transmission in slot $K$.
We can interpret the maximization in \re{eqn:J_k} as a tradeoff between the present and future rewards.
This is because the mutual information ${\mi}(\cdot,\cdot)$ represents the present reward, while $\bar{J}_{k+1}$, commonly known as the value function, is the expected future mutual information accumulated from slot $k+1$ until slot $K$.

Next, we obtain structural properties of the maximum throughput $\mathcal{T}^{\star}$  in \re{eqn:prob} and the corresponding optimal policy $\pi^{\star}$ in Theorems~\ref{thm:concave} and \ref{thm:monotone}. The proofs are given in the Appendix.

%\begin{proof}
%\end{proof}
%\overline{}
\begin{theorem}\label{thm:concave}
% \hl{Suppose $\mi(\snr, T)$ is concave in $T$ given $\snr$ and also concave in $T$ given $\snr$}.
Suppose that $\mi(\snr, T)$ is concave in $T$ given $\snr$. Given $\snr$ and $H$, then
%$J_k$ and $\bar{J}_k$ in Lemma~\ref{lem:1} satisfy
\begin{enumerate}
\item $J_k(\snr, H, B)$ in \re{eqn:J} is concave in $B$ for $k\in\mathcal{K}$;
\item $\bar{J}_k(\snr, H, B)$ in \re{eqn:J_expected} is concave in $B$ for $k\in\mathcal{K}$.
\end{enumerate}
Thus, $\mathcal{T}^{\star}=J_1(\state_1)$ is concave in $B_1$.
\end{theorem}

\begin{theorem}\label{thm:monotone}
%The optimal policy $\pi^{\star}$ in Lemma~\ref{lem:1} is monotonic in that
%:
%\begin{enumerate}
%\item $T^{\star}_k(\state_{k-1})$ is non-decreasing in $k\in\mathcal{K}$ if $\state_{k-1}=\state \; \forall k$.
%\item $T^{\star}_k(\state)$ is non-decreasing in $\snr$, $H$ and $B$, $k\in\mathcal{K}$.
%\item
Suppose that $\mi(\snr, T)$ is concave in $T$ given $\snr$.
Given $\snr$ and $H$, then the optimal power allocation $T^{\star}_k(\snr,H,B)$ that solves \re{eqn:J} is non-decreasing in $B$, where $k\in\mathcal{K}$.
%given the other two parameters in $\state$ is fixed
%\item $T^{\star}_k(\state)$ is non-decreasing in $\snr$ given $k, B$.
%\end{enumerate}
\end{theorem}

%The proofs for Theorems~\ref{thm:concave}, \ref{thm:monotone} can be obtained by suitable modification of the proofs in \cite{rockafellar70}, \cite{ross}.

The structural properties in Theorems~\ref{thm:concave} and \ref{thm:monotone} simplify the numerical computation of the optimal power allocation solution in Lemma~\ref{lem:1}, as shown in the next subsection.
%We can also show that $T^{\star}_k(\snr,H,B)$ is monotonic in $\snr$ or $H$ given further conditions of the the pdf \ref{eqn:distHsnr}, e.g. see \cite{ross}, \cite{Topkis}.

\subsection{Numerical Computations}

From \re{eqn:J_K1}, we get the optimal solution for slot $K$ as $T_K^{\star}(\state_{K})=B_K$.
Now, consider the problem of finding the optimal $T^{\star}_k(\state_k)$ to obtain $J_k(\state_k), k\in\{1,\cdots,K-1\}$.
%We can use Theorem~\ref{thm:concave} to obtain the optimal solution  $T_k^{\star}(\state_{k-1}), k<K$.
Let us fix the SNR and harvested energy as $\snr, H$, respectively, and drop these arguments when possible to simplify notations.
Consider the {\em unconstrained maximization} over all $T\geq0$, i.e., not subject to any energy harvesting constraint:
\be\label{eqn:num}
T_k^{\dagger}=\arg \max_{T\geq 0} g(T)
\ee
where we denote $g(T) = {\mi}(\snr, T)+ \bar{J}_{k+1}(B-T)$.
Since ${\mi}(\snr, T)$ is concave, and $\bar{J}_{k+1}(B-T)$ is concave due to Theorem~\ref{thm:concave}, the objective function $g(T)$ is concave. Thus, the maximization over all $T$ gives a unique solution $T_k^{\dagger}$, easily solved using numerical techniques such as a bisection search \cite{boyd}. Also, Theorem~\ref{thm:monotone} helps to reduce the search space by restricting the search to be in one direction for different $B$.
Alternatively, if $g(T)$ is differentiable and available in closed-form, $T_k^{\dagger}$ is given by solving $g'(T)=0$.
Finally, we get the optimal solution for \re{eqn:J_k} by restricting the maximization in \re{eqn:num} to be over $0\leq T\leq B$ to give
\be
T_k^{\star}=
\left\{
\begin{array}{ll}
0, & T_k^{\dagger}\leq 0; \\
B, & T_k^{\dagger}\geq B; \\
T_k^{\dagger}, & 0< T_k^{\dagger}< B.
\end{array}
\right .
\ee
This is because if $T_k^{\dagger}\leq 0,$ the (concave) objective function $g(T)$ must be decreasing for $T\geq 0$; if $T_k^{\dagger}\geq B,$ the objective function must be increasing for $T\leq B$.

\subsection{I.I.D. SNR and Harvested Energy}\label{subsec:iid}

We consider the {\em i.i.d. SI}  scenario where both $\snr_k$ and $H_k$ are independent and identically distributed (i.i.d.) over $k$ for analytical tractability.
% Hence, $\mathbb{E}_{{\snr_k}}[ \mi({\snr_k}, T)] \triangleq \bar{\mi}(T)$ is independent of the past channel state information $\snr_{k-1}$.
Even with i.i.d. SI, the optimization problem in Lemma~\ref{lem:1} is not decoupled as it still depends on the past harvested energy $H_{k-1}$.
Intuitively, this is because the present transmission energy $T_k$ (whose maximum allowable depends on $H_{k-1}$) will still affect the future storage energy $B_{k+1}, B_{k+2},\cdots$.

\newcommand{\snrave}{\bar{\snr}}
\newcommand{\Eone}[1]{\mathrm{E}_1\left(#1 \right)}

If we assume a {\em Rayleigh fading channel} with expected SNR  given by $\snrave$, i.e., the statistics of the SNR is $p_{\snr}(\snr)=1/\snrave \exp(-\snr/\snrave), \snr\geq 0,$ the expected mutual information evaluates as
\be\label{eqn:avemi_rayleigh}
\bar{\mi}( T)
\triangleq
\mathbb{E}_{{\snr}}[ \mi({\snr}, T)]
=\exp\left(\frac{1}{\snrave T}\right)  \Eone{\frac{1}{\snrave T}} \label{eqn:avemi_rayleigh1}
\ee
where the exponential integral is defined as $\Eone{x}=\int_x^{\infty}\exp(-t)/t \,\mathrm{d}t$.
Instead, if we assume an {\em AWGN channel} where the channel is time-invariant with $\snr_k=\snrave$ for all $k$, then the expected mutual information is simply
\be\label{eqn:avemi_awgn}
\bar{\mi}(T)&=& \mi(\snrave, T) =\log(1+\snrave T).
\ee
In AWGN channels, by inspection $J_k(\snr,H,B)$ in Lemma~\ref{lem:1} is independent of $\snr$ for all $k$, but still dependent on $H$. Hence, the optimization problem for  each $J_k(\snr,H,B)$ still has to be solved recursively, rather than as decoupled optimization problems.

% the  battery energy $B_k$ depends on the past transmission energy $T_{k-1}, T_{k-2}, \cdots,$ which in turn depends on the past $H_{k-1}, H_{k-2}, \cdots.$

%\begin{lemma}\label{lem:1}
%%Given $\state_1=(\snr_1, B_1)$, the maximum throughput is given by
%%\be\label{eqn:opt}
%%J_1(\snr_1, B_1)= \max_{\pi\in\Pi} \sum_{k=1}^K \mathbb{E}\left[\mi(\snr_k, T_k)|\snr_k, T_k, \pi\right].
%%\ee
%%where the expectation is performed over all random variables $H^K, \snr^K$ given $\snr_k, B_k$ using policy $\pi$.
%Given $\state_0=(\snr,H,B)$, the maximum throughput
%$\mathcal{T}^{\star}=J_1(\snr,H,B)$ can be computed recursively based on Bellman's equations, starting from $k=K$ until $k=1$ :
%\be\label{eqn:J}
%J_{K}(\snr,H,B) &=& \max_{0\leq T\leq B}
%\bar{\mi}(\snr, T)
%= \bar{\mi}(\snr, B) ,
%%\label{eqn:bellman}
%%\log(1+\snr T) = \log(1+\snr B),
%\IEEEyessubnumber
%\label{eqn:J_K1}\\
%J_k(\snr,H,B) &=& \max_{0\leq T\leq B} \bar{\mi}(\snr, T)
%%\log(1+\snr T)
%+ \bar{J}_{k+1}(\snr, H, B-T)
%\;\;\;\;\;\;
%\IEEEyessubnumber
%\label{eqn:J_k}
%\ee
%for $k=1,\cdots,K-1$,
%where
%\ben\label{eqn:I_expected}
%\bar{\mi}(\snr, T) &=& \mathbb{E}_{\tilde{\snr}}[ \mi(\tilde{\snr}, T)|\snr]  \\
% \bar{J}_{k+1}(\snr, H, x) &=& \mathbb{E}_{\tilde{H}, \tilde{\snr}}\left[ J_{k+1}(\tilde{\snr}, \tilde{H}, \min\{\Bmax, x+\tilde{H}\}) \big|\snr, H \right]
%\label{eqn:J_expected}
%\een

\section{Full Side Information: Arbitrary $\Bmax$}\label{sec:full-arbBmax}

The initial battery energy $B_1$ is always known by the transmitter.
We say that full SI is available if the transmitter also has priori knowledge of the harvest power  $H^{K-1}$ and SNR $\snr^K$ before any transmission begins.
%\footnote{$H_K$ is not needed in our problem, as the energy harvested in slot $K$ affects only the throughput for slot $K+1$ onwards.}.
This corresponds to the ideal case of a predictable environment where the harvest power and channel SNR are both known in advance, and also gives an upper bound to the maximum throughput $\mathcal{T}^{\star}$ for any distribution \re{eqn:distHsnr}.
%Moreover, it provides interesting insights that are useful for constructing practical schemes.

%\subsection{Arbitrary $\Bmax$}
In this section, we consider the general case where $\Bmax$ may be finite.
Corollary~\ref{cor:1}, as a consequence of Lemma~\ref{lem:1}, gives the optimal throughput $\mathcal{T}^{\star}$ for the same problem \re{eqn:prob} but with full SI available.
%All side information are {\it a priori} known and hence the SI is deterministic rather than random variables. Thus, \ref{eqn:distall} still holds with the pdfs replaced by Kronecker delta functions. By Applying \re{lem:1}, we immediately obtain Corollary~\ref{cor:1}.

\begin{corollary}\label{cor:1}
Given full SI $\{H^{K-1}, \snr^K\}$,
the maximum throughput is given by
\be\label{eqn:opt}
J_1(B_1)= \max_{\pi\in\Pi} \sum_{k=1}^K  \mi(\snr_k, T_k),
\ee
which can be computed recursively based on Bellman's equations:
\be\label{eqn:bellman}
&& J_{K}(B) = \max_{0\leq T\leq B}
\mi(\snr_K,  T) = \mi(\snr_K, B),
%\log(1+\snr T) = \log(1+\snr B),
\IEEEyessubnumber
\label{eqn:FSI:J_K1}\\
&& J_k(B) \nonumber \\
%\log(1+\snr T)
%\nonumber \\
&=& \max_{0\leq T\leq B} \mi(\snr_k, T) + J_{k+1}(\min\{\Bmax, B-T+H_{k}\})
\;\;\;\;\;\;\;\;
\IEEEyessubnumber
\label{eqn:FSI:J_k}
\ee
for $k=1,\cdots,K-1.$
%where
%\be\label{eqn:FSI:J_expected}
%\hat{J}_{k+1}(H_k,x) = J_{k+1}(\min\{\Bmax, x+H_{k}\}).
%\ee
\end{corollary}
\begin{proof}
All side information are {\it a priori} known and hence the SI is deterministic rather than random. Corollary~\ref{cor:1} thus follows immediately from Lemma~\ref{lem:1}, by replacing the pdfs in \re{eqn:distall} by Dirac delta functions accordingly.
%The proof follows from \ref{elm:1} by substitution
%Given FSI, the joint pdf \re{eqn:distHsnr} in \ref{elm:1}  consists of products of Dirac-delta functions
%\be
%\ee
\end{proof}

In general, power may be allocated via these modes:
\bi
\item {\em greedy} ({\sf \small G}): use all stored energy whenever available;
\item \emph{conservative} ({\sf \small C}): save as much stored energy as possible (without wasting any harvested energy) to the last slot;
\item \emph{balanced} ({\sf \small B}): stored energy is traded among slots accordingly to channel conditions.
\ei
%If $K=3$, we use the notation {\sf \small XY}, where {\sf \small X}, {\sf \small Y}~$\in$~\{{\sf \small G},~{\sf \small C},~{\sf \small B}\}, to suggest that in the first slot the ${\sf \small X}$ mode is used while in the second slot the {\sf \small Y} mode is used.

For the last slot, or if $K=1$ where there is only one slot, from \re{eqn:FSI:J_K1} it is optimal to allocate all power for transmission.
For the case $K=2$, Corollary~\ref{cor:K=2} obtains the optimal power allocation for the first slot. The proof is given in the Appendix.

\begin{corollary}\label{cor:K=2}
Consider $K=2$ slots. Suppose that the mutual information function is given by \re{eqn:gaussianmi}. Given full SI $\{B_1, H_1, \snr_1, \snr_2\}$, the optimal transmission energy for slot $1$ is given by (corresponding to the {\sf \small G}, {\sf \small B}, {\sf \small C} modes, respectively)
%\be\label{eqn:cor:K=2a}
%T_1^{\star}= B_1
%\ee
%if $H_1> \Bmax$, and is given by
\be\label{eqn:Topt}
T_1^{\star}=
\left\{
\begin{array}{lll}
B_1, & a<0 \mbox{ or } B_1<b ; \\
\widetilde{T}, & a\geq 0 \mbox{ and } -b\leq B_1 \leq c; \\
% |b|\leq B_1 \leq c  ; \\
\hspace{0cm}
[B_1-a]^+, & a\geq 0 \mbox{ and } (B_1>c \mbox{ or } B_1<-b);
\end{array}
\right .
%\\
%\;\hspace{-1cm}\label{eqn:cor:K=2}
\ee
%if $ H_1\leq\Bmax$,
where we denote $[x]^+\triangleq \max(0,x)$ and
\be\label{eqn:tildeT}
\widetilde{T}&=&  B_1/2 + (1/\snr_2-1/\snr_1+H_1)/2,
\ee
and we also let $a= \Bmax-H_1$,  $b= H_1+1/\snr_2-1/\snr_1$, and $c= 2 \Bmax-H_1+1/\snr_2-1/\snr_1.$
%and $b\triangleq H_1+(1/\snr_2-1/\snr_1), c\triangleq 2 \Bmax-H_1+(1/\snr_2-1/\snr_1).$
\end{corollary}

In Corollary~\ref{cor:K=2}, the power allocation  \eqref{eqn:Topt} is interpreted to be in {\sf \small G}, {\sf \small B}, or {\sf \small C} mode, respectively.
As an example, suppose $b>0$. Then all modes can be active: power allocation is greedy if the energy to be harvested is large or the stored energy is small ($a<0$ or $B_1<b$); power allocation is conservative if the energy to be harvested is small \emph{and} the stored energy is large ($a\geq 0$ and $B_1>c$); otherwise, the allocation depends on the SI.

%Moreover, as intuitively expected, a greedy mode is optimal if the initial battery energy is small enough, while a conservative mode is optimal if the initial battery energy is large enough.

\begin{remark}
From Corollary~\ref{cor:K=2}, $T_1^{\star}(B_1)$ is a piece-wise linear function of $B_1$. We also see that $T_1^{\star}(B_1)$ is increasing in $B_1$, as stated in Theorem~\ref{thm:monotone} for the general case.
\end{remark}

\begin{remark}\label{rem:powerhalf}
If $\Bmax\rightarrow\infty$, we get $a\geq 0$ and $c\rightarrow\infty$ in Corollary~\ref{cor:K=2}. Then \re{eqn:Topt} simplifies to $T_1^{\star}=[\widetilde{T}]^+$.
From \re{eqn:tildeT}, the optimal power allocation is thus given by half of the battery energy $B_1/2$ plus (or minus) a correction term that depends on the SNRs and harvested energy; this observation will be exploited to obtain a heuristic scheme in Section~\ref{sec:heu}.
\end{remark}

%\begin{remark}
%\hl{To check}
%From Corollary~\ref{cor:K=2}, we can recover the water-filling solution \cite{Goldsmith97IT} if we set $H_1=0$, for the case of $K=2$.
%\end{remark}

Although we can derive a closed-form result for larger $K$, the expression becomes unwieldy and less intuitive. To make progress, in the next section we assume the case of infinite $\Bmax$, which gives the highest possible achievable throughput and thus provides an upper bound for any practical implementation. The assumption is also reasonable if the storage buffer is selected to be large enough.
We shall show that for any $K$ we can obtain a closed-form result that is a variation of the water-filling power allocation policy \cite{Goldsmith97IT}, which is somewhat suggested by Remark~\ref{rem:powerhalf}.

%\newpage

\section{Full Side Information: Infinite  $\Bmax$}\label{sec:full-infBmax}

The previous section considers the general case of arbitrary $\Bmax$.
To develop more insights, in this section we consider that the mutual information function is given by \re{eqn:gaussianmi} and $\Bmax\rightarrow\infty$. Then from  \re{eqn:battery1}, the battery stored at slot $k+1$, where $k\in\mathcal{K},$ is given by
\be\label{eqn:Binf}
B_{k+1} =B_1-\sum_{i=1}^k T_i +\sum_{i=1}^k H_i.
\ee
A non-negative power allocation is feasible if and only if $B_{k+1} \geq 0, k\in\mathcal{K}$.
The throughput maximization problem solved in Corollary~\ref{cor:1} can then be formulated as follows:
\be\label{eqn:Binf:prob}\IEEEyessubnumber
\mathcal{T}^{\star}  =\max_{\{T_k\geq 0, k\in\mathcal{K}\}} %\mathcal{T}(\pi) &=&
&&
%\mathcal{T}(T^K) =
\sum_{k=1}^K \mi(\snr_k, T_k) \\
\IEEEyessubnumber\label{eqn:Binf:cons}
\mbox{ subject to }&& \sum_{i=1}^k T_i - B_1 -\sum_{i=1}^{k-1} H_i \leq 0, \;\;  k\in\mathcal{K}.
\IEEEeqnarraynumspace
\ee
%
%\be\label{eqn:Binf:prob}\IEEEyessubnumber
%\max_{\{T_k\geq 0, k\in\mathcal{K}\}} %\mathcal{T}(\pi) &=&
%\mathcal{T}(T^K) = \sum_{k=1}^K \mi(\snr_k, T_k)
%\ee
%subject to $T_k\leq B_k, k\in\mathcal{K}$, or equivalently, subject to
%\be
%\setcounter{IEEEsubequation}{2}
%\IEEEyessubnumber\addtocounter{equation}{-1}\addtocounter{IEEEsubequation}{1}\TexifyPS.edt
%\label{eqn:Binf:cons}
%\sum_{i=1}^k T_i &\leq& B_1 +\sum_{i=1}^{k-1} H_i, \;\;\;\; k\in\mathcal{K}
%\ee

\newcommand{\ToptWF}{\mathcal{T}^{\star}_{\mathsf{WF}}}

\subsection{Water-Filling Algorithm}

Before we consider the general case where the constraint \re{eqn:Binf:cons} is imposed for all $k\in\mathcal{K}$, we impose the constraint \re{eqn:Binf:cons} only for the last slot, i.e., only for $k=K$.
This then corresponds to the conventional problem of maximizing the sum throughput with a {\em sum} energy constraint of $P_{\max}= B_1 + \sum_{i=1}^{K-1} H_i$:
\be\label{eqn:Binf:prob:wf}\IEEEyessubnumber\label{eqn:Binf:prob:wf:obj}
\ToptWF(\snr^K, P_{\max})=
\max_{\{T_k\geq 0, k\in\mathcal{K}\}} %\mathcal{T}(\pi) &=&
&&
%\mathcal{T}(T^K) %=
\sum_{k=1}^K \mi(\snr_k, T_k)
\\
\IEEEyessubnumber\label{eqn:Binf:cons:wf}
\mbox{ subject to }&& \sum_{i=1}^K T_i \leq P_{\max}.
\ee
Since less constraints are imposed, the maximum throughput in \re{eqn:Binf:prob:wf} is no smaller than that of \re{eqn:Binf:prob}.
%In the context where energy is harvested over the slots, this means that all the energy harvested in slot $k\in\mathcal{K}$ are assumed to be available before slot $1$ is used for transmission.
It is well known that the optimal solution for \re{eqn:Binf:prob:wf} is given by (see e.g. \cite{Cover06,boyd})
\be\label{eqn:optP:wf}
{T}_{\mathsf{WF},k}^{\star}=\left[\nu-\frac{1}{\snr_k}\right]^+.
\ee
This optimal solution is implemented by the {\em water-filling algorithm}, where the {\em water-level} (WL) $\nu\geq 0$ is chosen such that \re{eqn:Binf:cons:wf} holds with equality by using the optimal power allocation in \re{eqn:optP:wf}.
For completeness, an implementation of the {\em water-filling algorithm}, which gives the maximum $\mathcal{T}$ to within a tolerance of $\epsilon$, is given below as Algorithm~\ref{alg:wf}.

%\renewcommand{\baselinestretch}{1}
%\IncMargin{1em}
\begin{algorithm}
\SetKwInOut{Input}{input}\SetKwInOut{Output}{output}
\Input{slot size $K$, SNRs $\{\snr_k\}$, power constraint $P_{\max}$, tolerance $\epsilon$ (close to zero)}
\Output{optimal power allocation $\{{T}_{\mathsf{WF},k}^{\star}\}$, optimal WL $\lambda^{\star}$}
\BlankLine
\tcp{initialization}
$P:=0$, $\lambda^{\mathrm{lo}}:=0, \lambda^{\mathrm{hi}}:=\infty$ (a large number) \;
$T^{\theta}_k:=[\lambda^{\theta}-1/\snr_k]^+, {\theta}\in\{\mathrm{lo},\mathrm{hi}\}, k\in\mathcal{K}$ \;
%$P^{\theta}:= \sum_{k=1}^K T^{\theta}_k, {\theta}\in\{\mathrm{lo},\mathrm{hi}\}$ \;
%$\mathcal{T}^{\theta}:= \sum_{k=1}^K \mi(\snr_k, T^{\theta}_k), {\theta}\in\{\mathrm{lo},\mathrm{hi}\}$ \;
\BlankLine
\tcp{loop until sum power $P$ less than $P_{\max}$ to within tolerance $\epsilon$}
\While{$|P_{\max}-P|>\epsilon$ or $P>P_{\max}$}
{
\tcp{improve $P$ to be closer to $P_{\max}$}
$\lambda :=(\lambda^{\mathrm{lo}}+\lambda^{\mathrm{hi}})/2$ \;
$T_k:=[\lambda-1/\snr_k]^+, k\in\mathcal{K}$ \;
$P := \sum_{k=1}^K T_k$\;
%$\mathcal{T}:= \sum_{k=1}^K \mi(\snr_k, T_k)$ \;
\tcp{update $\lambda_{\mathrm{lo}}$ or $\lambda_{\mathrm{hi}}$}
%\tcp{update $(\lambda_{\mathrm{lo}}, P_{\mathrm{lo}})$ or $(\lambda_{\mathrm{hi}}, P_{\mathrm{hi}})$}
\eIf{$P>P_{\mathrm{max}}$}
{
%\tcp{update $P_{\mathrm{hi}}$ as $P$}
$\lambda^{\mathrm{hi}}:=\lambda$\;
%$P^{\mathrm{hi}} := P$\;
}{
%\tcp{update $P_{\mathrm{lo}}$ as $P$}
$\lambda^{\mathrm{lo}}=\lambda$\;
%$P^{\mathrm{lo}} := P$\;
}
}
${T}_{\mathsf{WF},k}^{\star}:=T_k$\;
$\lambda^{\star}:=\lambda$\;
\caption{Conventional water-filling algorithm. This implementation achieves optimality to a tolerance of $\epsilon$.}\label{alg:wf}
\end{algorithm}
%\DecMargin{1em}
\renewcommand{\baselinestretch}{\mystretch}

\subsection{Staircase Water-Filling Algorithm}

We now proceed to solve our original problem \re{eqn:Binf:prob} with additional energy harvesting constraints in \re{eqn:Binf:cons}. It turns out that the conventional water-filling algorithm is no longer optimal. Instead it is necessary to use a generalized type of water-filling where the water level is a staircase-like function.

\subsubsection{Structural Properties}
The optimization problem in \re{eqn:Binf:prob} is convex and so can be solved by the dual problem \cite{boyd}.
The Lagrangian associated to the primal problem \re{eqn:Binf:prob} is
\ben
\mathcal{L}(\lambda^K, T^K)=
\mathcal{T}(T^K) - \sum_{k=1}^K \lambda_k  \cdot\left(\sum_{i=1}^k T_i - B_1- \sum_{i=1}^{k-1} H_i \right)
%+ \sum_{k=1}^K\mu_k T_k
\een
% \nonumber \\
%&=& \mu_K B_1 + \sum_{k=1}^K \mi(\snr_k, T_k) - \mu_k \cdot(T_k - H_k) \\
%&\triangleq & \mathcal{L}'(\{\mu_k\}, \{T_k\})
%\een
where  $T_k\geq 0$ is the power allocation for the $k$th slot and $\lambda_k\geq 0$ is the Lagrangian multiplier for the $k$th constraint in \re{eqn:Binf:cons}, $k\in\mathcal{K}$.
 %and  $\mu_k$ is the Lagrangian multiplier for the constraint $T_k\geq 0$ with $\lambda_k,\mu_k\geq 0$.
%where $\mu_k\triangleq \sum_{i=1}^k \lambda_i, k\in\mathcal{K}$ are used to transform the original Lagrange multipliers $\{\lambda_k\}$ to $\{\mu_k\}$ (and vice versa).
Then the necessary and sufficient conditions for $\lambda^K$ and $T^K$ to be both primal and dual optimal are given by the Karush-Kuhn-Tucker (KKT) optimality conditions:
 %(we denote the optimal solutions with a superscript $^{\star}$):
\be\label{eqn:KKT}
\IEEEyessubnumber\label{eqn:KKT:primalfeasible1}
\sum_{i=1}^k T_i -B_1 -\sum_{i=1}^{k-1} H_i & \leq& 0,
\\
\IEEEyessubnumber\label{eqn:KKT:primalfeasible2}
T_k &\geq & 0,
\\
\IEEEyessubnumber\label{eqn:KKT:dualfeasible}
\lambda_k &\geq& 0  ,
\\
\IEEEyessubnumber\label{eqn:KKT:CSC}
\lambda_k \left(\sum_{i=1}^k T_i  - B_1 -\sum_{i=1}^{k-1} H_i \right) &=&0,
\\
\IEEEyessubnumber\label{eqn:KKT:stationary}
\frac{\partial \mathcal{L}(\lambda^{K},  T^{K})} {\partial T_k}%\right|_{\lambda_k=\lambda_k,  T_k=T_k, k\in\mathcal{K}}
&=& 0 ,
\ee
for $k\in\mathcal{K}$.
%\be\label{eqn:KKT}
%\IEEEyessubnumber\label{eqn:KKT:primalfeasible1}
%\sum_{i=1}^k T^{\star}_i -B_1 -\sum_{i=1}^{k-1} H_i & \leq& 0,
%\\
%\IEEEyessubnumber\label{eqn:KKT:primalfeasible2}
%T^{\star}_k &\geq & 0,
%\\
%\IEEEyessubnumber\label{eqn:KKT:dualfeasible}
%\lambda^{\star}_k &\geq& 0  ,
%\\
%\IEEEyessubnumber\label{eqn:KKT:CSC}
%\lambda^{\star}_k \left(\sum_{i=1}^k T^{\star}_i  - B_1 -\sum_{i=1}^{k-1} H_i \right) &=&0,
%\\
%\IEEEyessubnumber\label{eqn:KKT:stationary}
%\left. \frac{\partial \mathcal{L}(\lambda^{K},  T^{K})} {\partial T_k}\right|_{\lambda_k=\lambda^{\star}_k,  T_k=T^{\star}_k, k\in\mathcal{K}} &=& 0 ,
%\ee
%for $k\in\mathcal{K}$.
%
From \re{eqn:KKT:primalfeasible2}, and imposing the constraints \re{eqn:KKT:dualfeasible} and \re{eqn:KKT:stationary} via similar arguments to obtain \re{eqn:optP:wf} in  \cite{Cover06,boyd}, we obtain the optimal power allocation as
\be\label{eqn:optP:gwf}
T^{\star}_k=\left[\nu_k-\frac{1}{\snr_k}\right]^+
\ee
for $k\in\mathcal{K}$, where $\nu_k\triangleq \left(\ln 2\sum_{i=k}^K \lambda_i \right)^{-1}\geq 0 $, and the $\lambda_i$'s satisfy the KKT conditions \eqref{eqn:KKT}.

Analogous to the problem in \re{eqn:Binf:prob:wf} with only power constraint \re{eqn:Binf:cons:wf}, we say $\nu_k$ is the WL for slot $k$.
Also, we say slot $t\in\mathcal{K}$ is a  {\em transition slot} (TS) if the water level changes {\em after} slot $t$, i.e.,  $\nu_{t}\neq \nu_{t+1}$.
We define the last slot $k=K$ also as a TS (say by defining $\nu_{K+1}$ to be infinity); hence there is at least one TS.
We collect all TSs as the set $\setT=\{t_1, t_2,\cdots, t_{|\setT|}\}$, where $t_i<t_j$ for $i<j$ and $t_{|\setT|}=K$.
%Assuming  there are $|\setT|< K$ TSs, we label the $i$th TS as $t_i, 1\leq i\leq |\setT|$.

\begin{figure}%[f]
\centering
%\psfrag{x}{ slot $k$}
%\psfrag{y}{}
%\psfrag{w}{$T_k$}
%\psfrag{s}{$\snr_k^{-1}$}
%\psfrag{v}{$\nu_k$}
%\psfrag{t}{$t_1$}
%\psfrag{u}{$t_2$}
\includegraphics[scale=1.25]{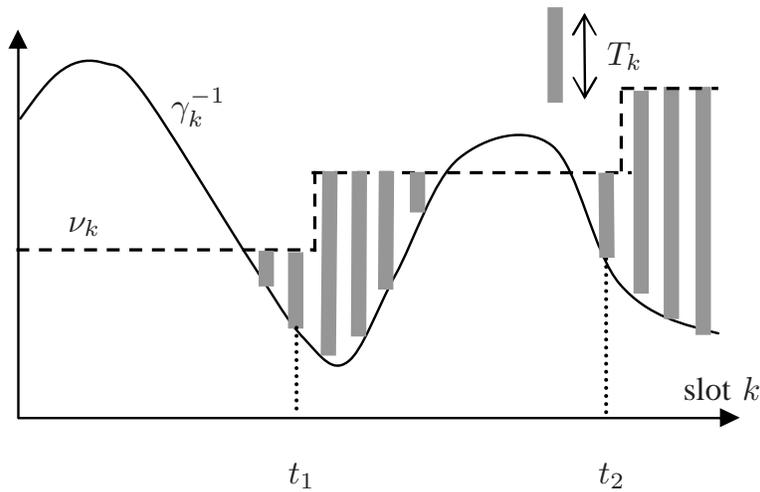}
\caption{Structure of optimal power allocation $T^{\star}_k$ with full SI and infinite $\Bmax$. We assume two optimal TSs and hence three distinct water levels for $\nu_k$. Here, the SNRs $\{\snr_k\}$ are arbitrary.}
\label{fig:wf}
\end{figure}

From the result in \re{eqn:optP:gwf}, we obtain the following structural properties for the optimal power allocation in Theorem~\ref{lem:L}.
%Fig.~\ref{fig:wf} gives an example of the optimal power allocation. It is possible to have multiple distinct water levels for $\nu_k$, but the water levels are non-increasing. In conventional water-filling, however, there is only one distinct water level.
Fig.~\ref{fig:wf} gives an example of the optimal power allocation.
In general the optimal WLs $\{\nu_k\}$ depend on the slot indices and there can be multiple optimal TSs, while in the conventional water-filling algorithm, the optimal WL is the same for all slot indices and thus there is no TS (except for the trivial one at slot $K$).

\begin{theorem}\label{lem:L}
%Suppose the mutual information function is given by \re{eqn:gaussianmi} and $\Bmax\rightarrow\infty$. Given full SI,
The optimal power allocation in \re{eqn:optP:gwf} satisfy these properties:
\bn
\item[$\mathsf{P}1$:] The WL is non-decreasing over slots, i.e.,  $\nu_1 \leq \cdots \leq \nu_K$.
We say that the optimal power allocation performs {\em staircase water-filling} over slots, since the WL is a staircase-like function (see e.g. Fig.~\ref{fig:wf}).
\item[$\mathsf{P}2$:] If slot $t$ is a TS, then the battery storage is empty, i.e., \re{eqn:KKT:primalfeasible1} holds with equality if $k\in \setT$.
\en
\end{theorem}
\begin{proof}
Since $\lambda_k\geq 0$, it follows that $\nu_k\geq 0$ and also that $\nu_k$ is non-decreasing with $k$.
This proves property $\mathsf{P}1$.

Suppose slot $t$ is an TS, i.e.,  $t\in \setT$ and so $\nu_{t}\neq \nu_{t+1}$.
Since by definition $\nu_k=(\ln 2\sum_{i=k}^K \lambda_i)^{-1}$, we get $\lambda_{t}\neq 0$. From \re{eqn:KKT:dualfeasible}, we get $\lambda_{t}>0$.  It then follows from the complementary slackness condition \re{eqn:KKT:CSC} (with $k$ replaced by $t$) that \re{eqn:KKT:primalfeasible1} holds with equality for $k=t$. This proves property $\mathsf{P}2$.
\end{proof}

From Theorem~\ref{lem:L}, we have the following additional structural properties in Corollary~\ref{rem:3} and Corollary~\ref{rem:piecewiseWF}.

\begin{corollary}\label{rem:3}
If the SNR is non-decreasing over slots, then the optimal power allocation is non-decreasing over slots.
\end{corollary}
\begin{proof}
This follows immediately from \re{eqn:optP:gwf} and property $\mathsf{P}1$, which implies that $T_l\leq T_k$ if $\snr_l\leq \snr_k$ for $l<k$.
\end{proof}

\begin{figure}%[f]
\centering
%\psfrag{x}{slot $k$}
%\psfrag{y}{}
%\psfrag{w}{$T_k$}
%\psfrag{s}{$\snr_k^{-1}$}
%\psfrag{v}{$\nu_k$}
\includegraphics[scale=1.25]{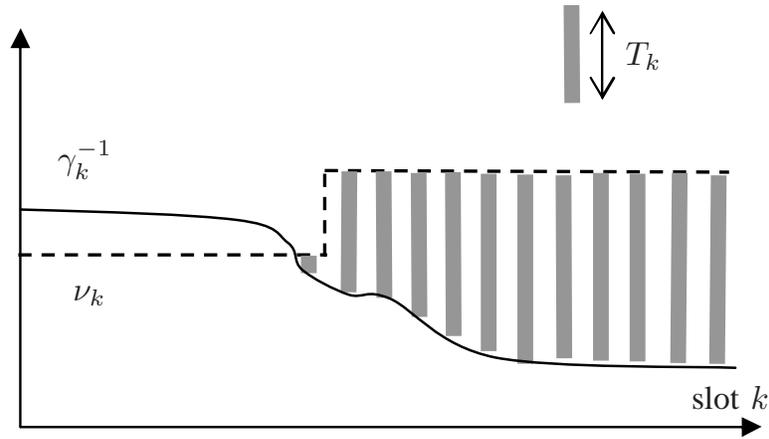}
\caption{Structure of optimal power allocation $T_k$ with full SI and infinite $\Bmax$, with increasing SNR $\snr_k$, i.e, decreasing $\snr_k^{-1}$, over slot $k$. In this case, $T^{\star}_k$ must increase over slot $k$.}
\label{fig:wf_incsnr}
\end{figure}

An example that illustrates Corollary~\ref{rem:3} is given in Fig.~\ref{fig:wf_incsnr}, where we see that the inverse of the SNR is non-increasing.
%Since the WL is non-decreasing from  $\mathsf{P}1$, clearly the optimal power allocation must be non-decreasing.
It is easy to see that the converse of Corollary~\ref{rem:3} is not true in general.
That is, if the SNR is non-increasing, then the optimal power allocation may not be non-increasing over slots (in particular for the slot immediately after the TS).
In conventional water-filling, however, both Corollary~\ref{rem:3} and its converse hold, i.e.,  the optimal power allocation is non-decreasing over slots {\em if and only if} the SNR is non-decreasing over slots.

We give an intuitive understanding of Corollary~\ref{rem:3}, and why the converse does not hold, to shed some light on how the energy harvesting constraints lead to a different optimal power allocation. First, let us consider the AWGN channel where the SNR is constant over slots.  If all the harvested energy is already available in the first slot, i.e., there is only a single sum-power constraint, a uniform power allocation is optimal for the AWGN channel. However, in an energy harvesting system, maintaining a uniform power allocation may not be always possible due to the causal arrival of the harvested energy.
% is always possible to accumulate energy so as to have more energy for transmission in the latter slots, but it may not be always possible to have more energy for transmission in the earlier slots.
Due to this non-uniform availability of harvested energy over slots, more energy only becomes available for transmission in the latter slots. Intuitively, we also expect more energy to be allocated for transmission in the latter slots such that the energy harvesting constraints in \eqref{eqn:Binf:cons} are satisfied.
This type of strategy is optimal from Corollary~\ref{rem:3}, which applies since the SNR is constant and hence also non-decreasing.
Next, consider the case where the SNR is non-decreasing over slots. From the water-filling algorithm under a single sum-power constraint, to achieve the maximum throughput it is optimal to allocate more power to the latter slots that have higher SNRs.  This is consistent with the earlier observation that more power should be allocated to the latter slots such that the constraints in \eqref{eqn:Binf:cons} are satisfied. Hence, it is also optimal to allocate more power to the latter slots.
In general if the SNR is arbitrary, however, the high-SNR slots may not correspond to the latter slots; hence intuitively the converse of Corollary~\ref{rem:3} may not hold in general.

%Finally, consider that the SNR is decreasing (or even arbitrary). Although the water-filling algorithm suggests to allocate more power to the earlier slots with higher SNR, this power allocation may not be always feasible due to the energy harvesting constraint. Hence in general one cannot conclude if the optimal power allocation is decreasing or increasing over slots.

\newcommand{\setSI}{\tilde{\setT}}

\subsubsection{Efficient Implementation}
Based on the structural properties, we now develop an efficient algorithm to implement the staircase water-filling.
%, the optimal power allocation has  the following interpretation in Corollary~\ref{rem:piecewiseWF}.

Some definitions are in order.
For convenience, let $t_0=0$.
We refer to the $i$th {\em slot interval}, where $i=1,\cdots,|\setT|$, as the slots between the $i$th and $(i+1)$th TS, specifically in the TS set $\setSI_i \triangleq \{t_{i-1}+1,\cdots, t_{i}\}$.
Thus, $\bigcup_i \setSI_i=\mathcal{K}$ and $\setSI_i \cap \setSI_j=\emptyset$ for $i\neq j$.
%Let the sum harvested energy that is available in the $i$th slot interval be $P(i) \triangleq  \sum_{k\in \setSI_i } H_{k-1}$, where we denote $H_0=B_1$ for notational simplicity.
The optimal set of TSs corresponding to an optimal power allocation is denoted as $\setT^{\star}=\{t^{\star}_1, t^{\star}_2, \cdots, t^{\star}_{|\setT^{\star}|}\}$

\begin{corollary}\label{rem:piecewiseWF}
The optimal power allocation performs staircase water-filling as follows:
 for every $i$th slot interval,  where $i=1,\cdots,|\setT^{\star}|$, conventional water-filling is performed  subject to the sum power constraint of
 $P(i) \triangleq  \sum_{k\in \setSI_i } H_{k-1}$, where we denote $H_0=B_1$ for notational simplicity.
 % and some constant WL, say  $\mu(i)$, where $\mu(i)$ in strictly increasing in $i$.
% Consider the $i$th interval  comprising of slot in after the $i$th TS until the $(i+1)$th TS where the optimal WL stays constant, where $0\leq i\leq |\setT|$. That is, $\mu_{t_i+1}=\mu_{t_i+2}=\cdots=\mu_{t_{i+1}}$.
\end{corollary}
\begin{proof}
%Consider the $i$th slot interval.
From property $\mathsf{P}2$,  all the harvested energy available in the $i$th slot interval, namely $P(i)$, is used during the $i$th slot interval. This follows by induction for $i=1,\cdots, |\setT|$. Moreover, the optimal power allocation in \re{eqn:optP:gwf} is equivalent to conventional water-filling.
To maximize throughput, the optimal power allocation must then be to use conventional water-filling with sum power constraint of $P(i)$ for every $i$th slot interval.
\end{proof}

%Given the optimal set of TSs $\setT^{\star}$, from Corollary~\ref{rem:piecewiseWF}  we can then obtain the optimal power allocation by performing a water-filling for each slot interval.
From Corollary~\ref{rem:piecewiseWF}, without loss of optimality the staircase water-filling solution comprises of multiple conventional water-filling solutions, one for each slot interval.
The original optimization problem \re{eqn:Binf:prob} can thus be reduced to a search for the optimal TS set $\setT^{\star}$ that has a size from $1$ to at most $K$:
% it is relatively straightforward to obtain the optimal power allocation (and concurrently the optimal WL) via a conventional water-filling algorithm.
%The optimization problem reduces to finding the optimal TSs:
\be\nonumber
\mathcal{T}^{\star}  =
\max_{1\leq |\setT|\leq K} \max_{\setT} % \max_{\setT\in \{1\leq t_1 < \cdots <t_{|\setT|}\leq K\}}
\;
& & \ToptWF(\snr_{1}^{t_1}, P(1)) + \ToptWF(\snr_{t_1+1}^{t_2}, P(2)) \\
&&+   \cdots + \ToptWF(\snr_{{t_{|\setT|}}-1}^{t_{|\setT|}}, P(|\setT|))
\label{eqn:optprob2}
\IEEEeqnarraynumspace
\ee
subject to the power allocation $P(1), \cdots, P(|\setT|)$ satisfying the constraints in \re{eqn:Binf:cons}.
A brute force search based on \eqref{eqn:optprob2} is of a high computational complexity.
%, since the size $|\setT^{\star}|$ of the optimal set of TSs is not known in advance.
Nevertheless, it turns out that it is optimal to simply employ a  forward-search procedure, starting with the search of the optimal $t_1^{\star}$, then of the optimal $t_2^{\star}$, and so on until the last optimal TS $t_{|\setT^{\star}|}^{\star}$ equals  $K$, at which point the optimal size $|\setT^{\star}|$ is also obtained.

The first optimal TS $t_1^{\star}$ can be found in Lemma~\ref{lem:2} given below; by induction, the search of the subsequent optimal TSs will follow similarly.
Lemma~\ref{lem:2} requires the following {\em feasible-search procedure} for a given optimization problem \re{eqn:Binf:prob}:
\begin{enumerate}
\item \label{step:1}
Initialize $\setT_1$ as an empty set.
\item \label{step:2}
For $t_1=1, \cdots, K$, obtain the optimal power allocation from slot $1$ to slot $t_1$ by  using a water-filling algorithm (such as Algorithm~\ref{alg:wf}) assuming that all harvested energy is available, i.e., the sum power constraint is $P_{\max}=B_1+\sum_{i=1}^{t_1-1} H_i$.
\item \label{step:3}
Admit $t_1$ in the set $\setT_1$ if the corresponding optimal power allocation satisfies the constraint \re{eqn:Binf:cons} for $k=1,\cdots,t_1$.
\end{enumerate}
%Any $t_1\in\setT_1$ is a feasible solution to the optimization problem \re{eqn:Binf:prob} due to the admission check in Step~\ref{step:3}.
%The feasible-search procedure produces a {\em feasible} set  $\setT_1$ for the first TS, since any $t_1\in \setT_1$ satisfies the constraints to the optimization problem  \re{eqn:Binf:prob} due to Step~\ref{step:2}.
The set $\setT_1$ is non-empty; it contains at least the element $t_1=1$, as the constraint \re{eqn:Binf:cons} in Step~\ref{step:3} is equivalent to the sum power constraint in Step~\ref{step:2}.
Moreover, the set $\setT_1$ includes all possible candidates for the optimal $t_1^{\star}$; the only candidates that are not included are those where at least one of the constraint in \re{eqn:Binf:cons}  is not satisfied for slot $k=1,\cdots, t_1^{\star}$.

\begin{lemma}\label{lem:2}
Let $\setT_1$ be the feasible set of $t_1$ obtained by the feasible-search procedure.
Then the optimal TS is given by the largest element in $\setT_1$, i.e.,
\be\label{eqn:optTP}
t_1^{\star} &=&\max_{t_1\in \setT_1} t_1.
\ee
\end{lemma}
\begin{proof}
%The optimal power allocation is obtained by staircase water-filling from Corollary~\ref{rem:piecewiseWF} .
If $|\setT_1|=1$, then that only $t_1$ must be optimal.
Henceforth, assume $|\setT_1|\geq 2$.
Consider two TSs $t', t''\in \setT_1$, where $t'<t''$. Denote their respective optimal WLs obtained from the water-filling algorithm as  $\nu', \nu''$. Then $\nu' \geq \nu''$. Otherwise if $\nu' < \nu''$, more power is allocated for each time slot $k=1\cdots,t',$ if the WL $\nu''$ is used, compared to the case where the WL $\nu'$ is used. But since the power allocation with WL $\nu'$ has used all available power at slot $t'$ due to property $\mathsf{P}2$ in Theorem~\ref{lem:L}, the power allocation with WL $\nu''$ is infeasible and thus cannot be optimal.

We now show that $t'$ cannot be the optimal $t_1^{\star}$ by contradiction. Suppose that $t_1^{\star}=t'$, i.e., water-filling is used from slot $1$ to slot $t'$ with WL $\nu'$. The WL for the power allocation must then subsequently decrease at some slot $t'< k\leq t''$, otherwise the sum power allocated from slots $1$ to slot $t''$ will be more than the sum power allocated with the (constant) WL $\nu''$,  which violates the sum power constraint. But from property $\mathsf{P}1$ in Theorem~\ref{lem:L}, the optimal WL is non-decreasing. Thus, $t_1^{\star}\neq t'$ by contradiction. By induction, all elements in $\setT_1$, except for the largest one, are suboptimal. The only candidate left, namely the largest element, must then be optimal.
\end{proof}

We now propose Algorithm~\ref{alg:tp} below to solve \re{eqn:optprob2}, which is optimal according to Theorem~\ref{thm:optalgo} below.
Briefly, Algorithm~\ref{alg:tp} performs a forward-search procedure (starting from slot $1$) in each of the outer iteration for $t_1^{\star}, t_2^{\star},\cdots,$ until $t_{|\setT^{\star}|}^{\star}=K$. Given that $t_{1}^{\star}, \cdots, t_{i-1}^{\star}$ is found, to obtain $t_i^{\star}$, an inner iteration is performed via a {\em backward-search} procedure, starting from slot $t_{i-1}^{\star}, t_{i-1}^{\star}-1,\cdots,$ until slot $1$.

\renewcommand{\baselinestretch}{1}
%\IncMargin{1em}
\begin{algorithm}
\SetKwInOut{Input}{input}\SetKwInOut{Output}{output}
\Input{slot size $K$; SNRs $\{\snr_k\}$; harvested power $\{H_k\}$ where we let $H_0=B_1$;  tolerance $\epsilon$ %(close to zero)
}
\Output{optimal set of TSs $\setT^{\star}=\{t_i^{\star}\}$} %, optimal WL $\lambda$}
\BlankLine
\tcp{initialization}
$t_0=0$\;

\BlankLine
\For{$i=1, 2, \cdots, K$}
{   \tcp{Outer iteration: find $t_1^{\star}$, then $t_2^{\star}$, and so on}
    \For{$k=K, K-1,\cdots,t_{i-1}+1$}
    {
    \tcp{Inner iteration: find the largest feasible $t_i$ in \re{eqn:optTP}}
    Use Algorithm~\ref{alg:wf} for slot $t_{i-1}+1$ to slot $k$ with inputs \\
    (i) SNRs $\{\snr_{t_{i-1}+1}, \cdots, \snr_k\}$ \\
    (ii) $P_{\max}=\sum_{i=t_{i-1}}^{k-1} H_{i}$ \\
    (iii) tolerance $\epsilon$, \\
    to give output  $\{T^{\star}_{t_{i-1}+1}, \cdots, T^{\star}_k\}$;
    \\
    \If{$\{T^{\star}_1, \cdots, T^{\star}_k\}$ satisfy the constraints \re{eqn:Binf:cons}}
    { $t_i^{\star}:=k $;}
    %\Else{$t_i(k):=\emptyset $;}
    %WF $\nu_k$ by ignoring whether the harvested constraint is satisfied or not\;
    %reject $\nu_k$ if constraint is not satisfied\;
    }
%    {$t_i(k):=k $;}
%    $t_i^{\star}:=\arg\max_{k} t_i(k)$\;
    %$\nu_i:= \min_{k} \nu_k$\;
    \If{$t_i^{\star}=K$ }{\bf{exit}\;}
}

%\While{$|P_{\max}-P|>\epsilon$ or $P>P_{\max}$}
%{
%\tcp{update $P$ to be closer to $P_{\max}$}
%$\lambda :=(\lambda^{\mathrm{lo}}+\lambda^{\mathrm{hi}})/2$ \;
%$T_k:=[\lambda-1/\snr_k]^+, k\in\mathcal{K}$ \;
%$P := \sum_{k=1}^K T_k$\;
%%$\mathcal{T}:= \sum_{k=1}^K \mi(\snr_k, T_k)$ \;
%\tcp{update $(\lambda_{\mathrm{lo}}, P_{\mathrm{lo}})$ or $(\lambda_{\mathrm{hi}}, P_{\mathrm{hi}})$}
%\eIf{$P>P_{\mathrm{max}}$}
%{
%%\tcp{update $P_{\mathrm{hi}}$ as $P$}
%$\lambda^{\mathrm{hi}}:=\lambda$\;
%$P^{\mathrm{hi}} := P$\;
%}{
%%\tcp{update $P_{\mathrm{lo}}$ as $P$}
%$\lambda^{\mathrm{lo}}=\lambda$\;
%$P^{\mathrm{lo}} := P$\;
%}
%}
\caption{Finding optimal TSs}\label{alg:tp}
\end{algorithm}
%\DecMargin{1em}
\renewcommand{\baselinestretch}{\mystretch}

%We give an algorithm to obtain the optimal TS in Algorithm~\ref{alg:tp}. It makes use of the conventional water-filling algorithm.

\begin{theorem}\label{thm:optalgo}
Algorithm~\ref{alg:tp} obtains the optimal $\setT^{\star}$ that solves the optimization problem in \re{eqn:optprob2} or equivalently \re{eqn:Binf:prob}.
\end{theorem}
\begin{proof}
The $i$th outer iteration of Algorithm~\ref{alg:tp} finds the optimal $t_i^{\star}$. Consider $i=1$.
From Lemma~\ref{lem:2}, we can determine the optimal $t_1^{\star}$ by finding the largest feasible $t_1$. Without loss of optimality, we can modify the feasible-search procedure such that the search (in Step~\ref{step:2}) starts from the largest slot index to the smallest, and (Step~\ref{step:3}) terminates to give the optimal $t_1^{\star}$ once a feasible $t_1$ is found. These modifications lead to the inner iteration in  Algorithm~\ref{alg:tp}.

From Property~$\mathsf{P}2$ in Theorem~\ref{lem:L}, in the first TS interval all the power available would be used. Since no power is available for subsequent slot intervals given $t_1^{\star}$, the power allocation for subsequent slot intervals can be optimized independent of the actual power allocated in the first slot interval.
The throughput maximization problem from slot $t_1^{\star}+1$ onwards can be solved similarly as before (after removing time slots $1, \cdots, t_1^{\star}$).
Thus, we apply the inner iteration again to determine $t_2^{\star}$, similarly for $t_3^{\star}$ and so on, as reflected in the outer iteration of Algorithm~\ref{alg:tp}. The iteration ends when the optimal TS equals $K$, which is the largest possible value as stated in optimization problem in \re{eqn:optprob2}.
\end{proof}

\subsubsection{Update Algorithm when New Slots Become Available}

Suppose that we have obtained the optimal $\setT^{\star}$ based on Algorithm~\ref{alg:tp} for a $K$-slot system. Now, a new slot becomes available for our use where its SI is known. We wish to obtain the new optimal solution for this $(K+1)$-slot system, say $\setT^{\star}_{\text{new}}=\{t^{\star}_{1,{\text{new}}}, t^{\star}_{2,{\text{new}}}, \cdots\}$. Instead of implementing Algorithm~\ref{alg:tp} afresh, we can obtain an update of  $\setT^{\star}_{\text{new}}$ from $\setT^{\star}$ as follows:
\begin{itemize}
\item  Consider $i=1$ in the outer iteration of Algorithm~\ref{alg:tp}. We only need to execute $k=K+1$ in the inner iteration. If the constraints \re{eqn:Binf:cons} are satisfied, then we have obtained $t^{\star}_{1,{\text{new}}}=K+1$, hence  $\setT^{\star}_{\text{new}}=\{t^{\star}_{1,{\text{new}}}\}$ and Algorithm~\ref{alg:tp} terminates. Otherwise, since we already know that  $t^{\star}_1$ in the largest element in $ \setT^{\star}$, we obtain immediately $t^{\star}_{1,{\text{new}}}=t^{\star}_{1}$.
\item The subsequent iterations are executed similarly. For the $i$th outer iteration, where $i=2,3,\cdots$, we only execute $k=K+1$ in the inner iteration. If the constraints \re{eqn:Binf:cons} are satisfied, then $t^{\star}_{i,{\text{new}}}=K+1$ and Algorithm~\ref{alg:tp} terminates; otherwise  $t^{\star}_{i,{\text{new}}}=t^{\star}_{i}$.
\end{itemize}
Hence, for every outer iteration, only one additional inner iteration is executed until the constraints \re{eqn:Binf:cons}  are satisfied. Since there are at most $K$ outer iterations, we need to execute at most  $K$ inner iterations in total.
In cases when the number of available slots can increase dynamically in a multi-user system, say when other users give up their slots and is assigned to our energy harvesting system, the above proposed update algorithm allows an efficient way to update  $\setT^{\star}$.

\section{Numerical Results}\label{sec:num}

% the graphs below do not use instantaneous CSI
%\begin{figure}[f]
%\centering
%\includegraphics[scale=\FigSize]{causal_full_si_awgn_v2a.eps}
%\caption{AWGN channel: optimal throughput when causal SI (blue with  ``$\times$'' markers) or full SI (red with ``$\circ$'' markers)  is available for $K=1,2,4$.}
%\label{fig:causal_full_si_v2a}
%\end{figure}
%
%\begin{figure}[f]
%\centering
%\includegraphics[scale=\FigSize]{causal_full_si_v4.eps}
%\caption{Rayleigh fading channel: optimal throughput when causal SI (blue with  ``$\times$'' markers) or full SI (red with ``$\circ$'' markers)  is available for $K=1,2,4$.}
%\label{fig:causal_full_si_v4}
%\end{figure}
%
%\begin{figure}%[f]
%\centering
%\includegraphics[scale=\FigSize]{full_si_awgn_v1.eps}
%\caption{AWGN channel: optimal throughput when full SI is available for large number of slots $K$.}
%\label{fig:causal_full_si_v1}
%\end{figure}
%
%
%\begin{figure}%[f]
%\centering
%\includegraphics[scale=\FigSize]{full_si_v1.eps}
%\caption{Rayleigh fading channel: optimal throughput when full SI is available for large number of slots $K$.}
%\label{fig:full_si_v1}
%\end{figure}

To obtain numerical results, we assume that the SNR $\snr_k$ and the harvested energy $H_k$ are i.i.d. over time slot $k$ and the channel is either an AWGN or Rayleigh fading channel, as described in Section~\ref{subsec:iid}.
We assume the initial stored energy $B_1$ and the harvested energy $H_k$ takes a value in $\{0,0.5,1\}$ with equal probability and $\Bmax\rightarrow\infty$.
To measure the performance of various schemes, we plotted the throughput per slot, i.e., the sum throughput divided by the number of slots $K$, as the average SNR $\snrave$ is increased.

%The results are obtained by numerical calculations or Monte Carlo simulations, as follows.

If causal SI is available, the optimal policy is obtained recursively by applying Lemma~\ref{lem:1}. Specifically, we first obtain  $\bar{J}_{K}(B)$ in \re{eqn:J_expected} via the closed-form results in Section~\ref{subsec:iid}; we drop the $\snr$ and $H$ arguments due to the i.i.d. assumption. Then we obtain  ${J}_{K-1}(\snr, B)$ in \re{eqn:J_k}, say by an iterative  bisection method. The throughput are averaged over $10^4$ independent realizations of $\snr$ and $H$ to give $\bar{J}_{K-1}(B)$.
This procedure is performed for different $B$, discretized in step size of $0.01$, and stored to be used for the next recursion. The iteration is repeated for $k=K-2, \cdots, 1$.
If instead full SI is available, the optimal policy is obtained by Algorithm~\ref{alg:tp}; we have verified that our proposed algorithm is significantly faster but is equivalent to solving the problem via a  standard optimization software. The throughput per slot is obtained from averaging the results from $10^4$ independent Monte Carlo runs.

The results are shown in Fig~\ref{fig:causalfull_si_awgn_instantchan_v1} for AWGN channels and in Fig.~\ref{fig:causalfull_si_fade_instantchan_v1} for Rayleigh fading channels.
The throughput in both cases, when either full SI or causal SI is available, is the same for $K=1$, because any SI cannot be exploited for future slots.
However, in both cases the throughput per slot increases as $K$ increases. The increment is more substantial when full SI is available, intuitively because the SI can then be much better exploited.
The incremental improvement as $K$ increases is significant when $K$ is small, but becomes less significant when $K$ is large.
The throughput with either full SI or causal SI does not differ significantly, possibly because the SI that can be further exploited from full SI is limited in our i.i.d. scenario.

\begin{figure}%[f]
\centering
\includegraphics[scale=\FigSize]{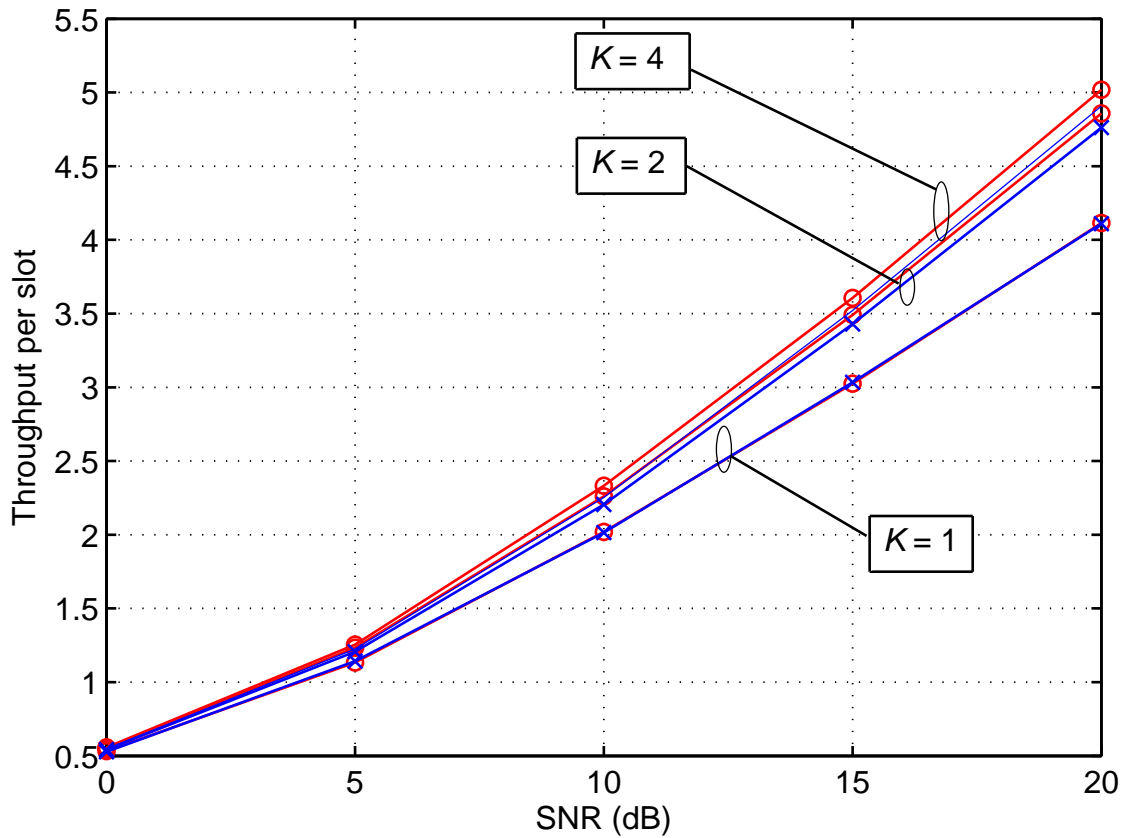}
\caption{AWGN channel: optimal throughput when causal SI (blue with  ``$\times$'' markers) or full SI (red with ``$\circ$'' markers)  is available for $K=1,2,4$.}
\label{fig:causalfull_si_awgn_instantchan_v1}
\end{figure}

\begin{figure}%[f]
\centering
\includegraphics[scale=\FigSize]{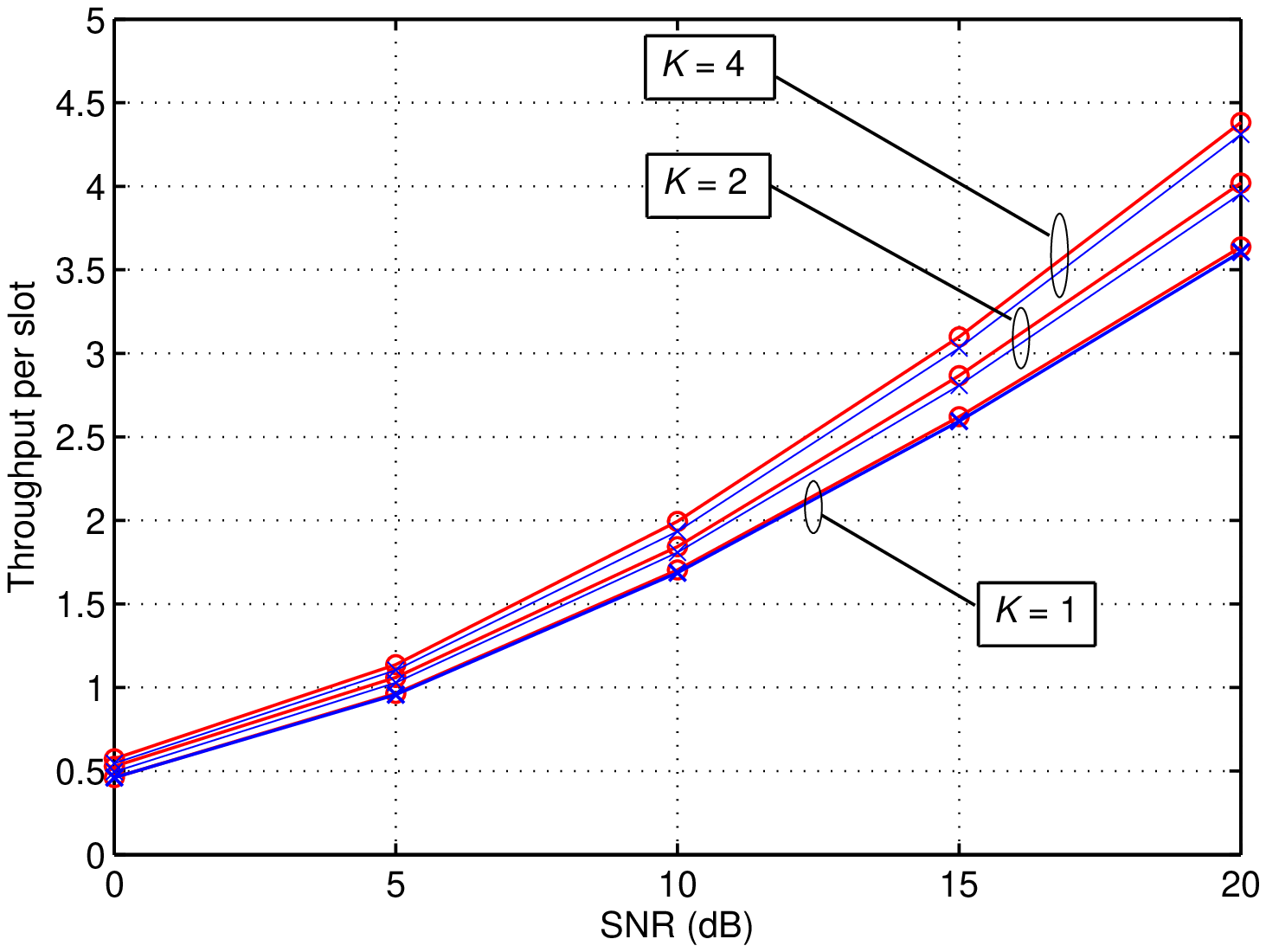}
\caption{Fading channel: optimal throughput when causal SI (blue with  ``$\times$'' markers) or full SI (red with ``$\circ$'' markers)  is available for $K=1,2,4$.}
\label{fig:causalfull_si_fade_instantchan_v1}
\end{figure}

\subsection{Heuristic Schemes with Causal SI}\label{sec:heu}

Next, we consider two heuristic schemes that use causal SI yet can be easily implemented in practice, namely the naive scheme and the power-halving scheme.

In the {\em naive scheme}, all stored energy $B_k$ is used in every slot $k$, i.e., $T_k=B_k$. This is equivalent to the case of $K=1$ in our optimization problem regardless of whether causal SI is available (see Lemma~\ref{lem:1}) or full SI is available (see Theorem~\ref{thm:optalgo}). In both cases, it is optimal to use all stored energy. As seen earlier, the case of $K=1$ performs significantly worse than the optimal schemes for $K>2$ in both cases. To obtain further improvement in the per-slot throughput, we need to further exploit the causal SI available.

\begin{figure}%[f]
\centering
\includegraphics[scale=\FigSize]{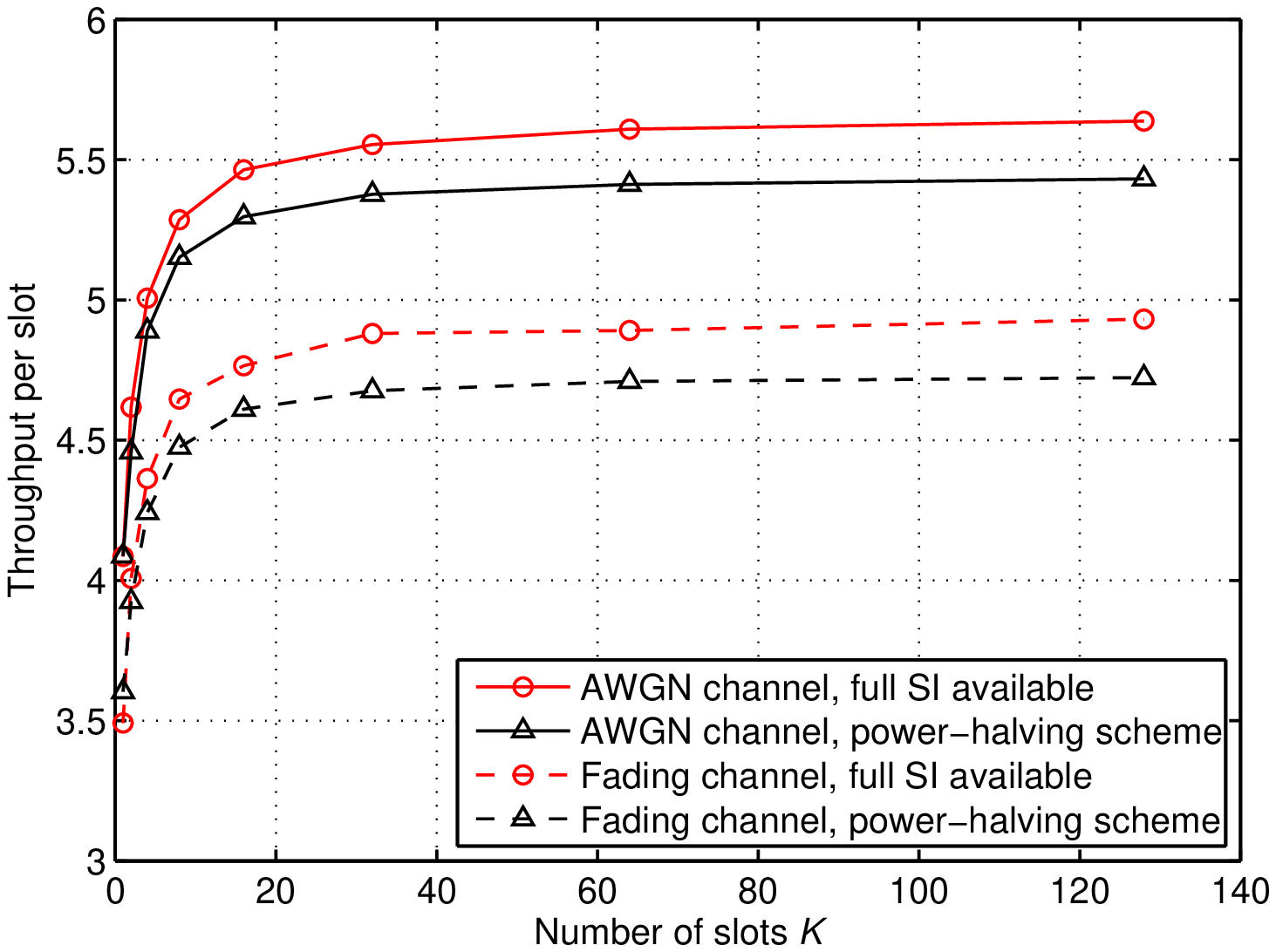}
\caption{Throughput based on the power-halving scheme. The optimal throughput when full SI is also plotted for comparison.}
\label{fig:half_batt_instanchan_diffK_v1}
\end{figure}

In the {\em power-halving scheme}, all stored power is used for transmission in the last slot, while for all other slots  half of the stored energy is used,  i.e., $T_k= w_k B_k$ where $w_k=1$ if $k=K$ and $w_k=1/2$ otherwise.
This scheme is simple to implement. Intuitively, the present throughput is traded equally with the future throughput by splitting the battery energy $B_k$ into two halves. We note that this scheme also implicitly exploits causal information of the harvested energy (which accumulates as the stored energy $B_k$).
Moreover, the power-halving scheme satisfies the following characteristics and thus likely improves the throughput when causal SI is only available:
\begin{enumerate}
\item $T_k$  increases with $B_k$, in accordance with Theorem~\ref{thm:monotone} in the causal SI case.
\item From Remark~\ref{rem:powerhalf}, if $K=2$, the optimal power allocation is given by half of the battery energy $B_1/2$ and a correction term that depends on the SNRs and harvested energy. Ignoring this correction term (which is not known due to the lack of full SI) then leads to the power-halving scheme for $K=2$. %Numerical results (see Fig.~\ref{fig:half_batt_instanchan_diffK_v1} in Section~\ref{sec:num})  also show that the power-halving scheme performs close to the optimal policy for small $K$.
    %This scheme corresponds to \re{eqn:tildeT}  in Corollary~\ref{cor:K=2} by setting the second term to be zero, in the case of full SI where $K=2$. That is, only the harvested energy that has been accumulated in the battery is used as SI.
\item More stored energy is deferred to be used in the latter slots, thus resembling the optimal policy with staircase WLs in the full SI case.
\end{enumerate}

Fig.~\ref{fig:half_batt_instanchan_diffK_v1} shows the throughput per slot obtained  by averaging the numerical results from $2\times 10^4$ independent runs of Monte Carlo simulations, for both AWGN channels and Rayleigh fading channels.
We fix the SNR at $20$~dB.
As benchmarks, we also plot the optimal throughput when full SI is available. This is because the computational complexity in solving the Bellman's equations in Lemma~\ref{lem:1} when causal SI is available becomes prohibitive for large $K$; moreover we observed earlier that the performance with causal SI is close to the performance with full SI.
The results show that the power-halving scheme is able to improve on the per-slot throughput as $K$ is increased, and is only within about $0.2$~bits away from the case when full SI is available.
Moreover, we see that at small $K$, the gap is even closer, as suggested by the second characteristic mentioned above.
Similar results are obtained at lower SNR, with an even smaller throughput degradation compared to the case when full SI is available.
%Similar to the optimal schemes, the increment in improvement is significant for small $K$ while the improvement is negligible for large $K$.
%The throughput per slot obtained is also comparable to the optimal scheme with causal SI. For example for the case of Rayleigh fading channel with $K=4$, by comparing Fig.~\ref{fig:halfbatt_fade_v1} with Fig.~\ref{fig:causal_full_si_v4} we see that the throughput degradation is less than $0.1$ bits for all SNRs.
Further performance gain may also be obtained by optimizing this tradeoff by considering the channel conditions explicitly.

%\begin{figure}[f]
%\centering
%\includegraphics[scale=\FigSize]{halfbatt_awgn_v1.eps}
%\caption{AWGN channel: throughput based on the power-halving scheme, where half of the stored energy is used for transmission energy (except for the last slot where all stored energy is used).}
%\label{fig:halfbatt_awgn_v1}
%\end{figure}
%
%\begin{figure}[f]
%\centering
%\includegraphics[scale=\FigSize]{halfbatt_fade_v1.eps}
%\caption{Rayleigh fading channel: throughput based on the power-halving scheme, where half of the stored energy is used for transmission energy (except for the last slot where all stored energy is used).}
%\label{fig:halfbatt_fade_v1}
%\end{figure}

\section{Conclusion}\label{sec:con}
We considered a communication system where the energy available for transmission varies from slot to slot, depending on how much energy is harvested from the environment and expended for transmission in the previous slot.
We studied the problem of maximizing the throughput via power allocation over a finite horizon of $K$ slots, given either causal SI or full SI.
We obtained structural results for the optimal power allocation in both cases, which allows us to obtain efficient computation of the optimal throughput. Finally, we proposed a heuristic scheme where numerical results show that the throughput per slot increases as $K$ increases and performs relatively well compared to a naive scheme.
%, since the optimal solution is usually complex to obtain.
%Given full SI, we obtain a closed-form solution for the case of $K=2$. We also obtain the structure of the solution for arbitrary $K$ but with unlimited energy storage available. The optimal solution has a water-filling interpretation, as in \cite{Goldsmith97IT}. However, instead of a single water level, there are multiple water levels that decrease over time.

%
\appendices

%\appendix
\section{Proof of Theorem~\ref{thm:concave}}
\label{append:a}
%\begin{proof}[Proof of Theorem~\ref{thm:concave}]
% From \re{eqn:J_k}, clearly $J_k(\snr, B)$ is concave in $\snr$ for all $k$.
%Let $\state=(\snr, H, B).$
With $\snr, H$ fixed, we prove by induction that $J_k(\snr, H, B)$ and $\bar{J}_{k+1}(\snr, H, x)$ are concave in $B$ and $x$, respectively, for decreasing
%decreasing $k$, i.e., for
$k=K, \cdots, 1$.

Consider $k\in\{1,\cdots, K-1\}$.
Suppose that $J_{k+1}(\snr, H, B)$ is concave in $B$.
We note that $J_{k+1}({\snr}, H, \min\{\Bmax, x+H\})$ is concave in $x$, as it is the minimum of $J_{k+1}({\snr}, H, \Bmax)$ (a constant independent of $x$) and the concave function $J_{k+1}({\snr}, H, x+H)$.
It follows that $\bar{J}_{k+1}(\snr, H, x)$ % in \re{eqn:J_expected}
is concave in $x$, since expectation preserves concavity.
From \re{eqn:J_k}, $J_k$ is a supremal convolution of two concave functions in $B$, namely ${\mi}$ and $\bar{J}_{k+1}$ (with $\snr, H$ fixed). It follows that $J_k$ is concave in $B$, since the infimal convolution of convex functions is convex \cite[Theorem 5.4]{rockafellar70}.
To complete the proof by induction, we note that $J_{K}(\snr, H, B)=\mi(\snr,B)$ % in \re{eqn:J_K1}
is concave in $B$ by assumption on the mutual information function $\mi(\cdot,\cdot)$.
%By induction $J_k(\snr, B)$ is concave in $B$ for decreasing $k=K,\cdots, 1$. This concludes the proof.
%\end{proof}

\section{Proof of Theorem~\ref{thm:monotone}}
%\section{An Auxiliary Lemma}
We need Lemma~\ref{lem:monotonic} to prove Theorem~\ref{thm:monotone}.
%This result is a specific case of \cite{Topkis}.
%; the proof is in \cite[Theorem~2]{Amir}.
\begin{lemma}\label{lem:monotonic}
Consider %the problem
$T^{\star}(B) = \arg \max F(B,T),$ where the maximization is over interval $T_{l}(B)\leq T \leq T_{u}(B)$ that depends on $B$.
If $T_{l}(B), T_{u}(B)$ are non-decreasing in $B$, and if $F$ has non-decreasing differences in $(B,T)$, i.e., $\forall T'\geq T, B'\geq B$,
\be\label{eqn:incdiff}
F(B',T')-F(B,T')\geq F(B',T)-F(B,T),
\ee
then the maximal and minimal selections of $T^{\star}(B)$, denoted as $\overline{T}(B), \underline{T}(B)$, are non-decreasing.
\end{lemma}
\begin{proof}
See proof in \cite[Theorem~2]{Amir}. %, a special case of \cite{Topkis}.
%See \cite{Amir}. %[Theorem~2]{Amir}.
\end{proof}

% \section{Proof of Theorem~\ref{thm:monotone}}
%\begin{proof}[Proof of Theorem~\ref{thm:monotone}]

We now prove Theorem~\ref{thm:monotone} with $\snr, H$ fixed; we drop these arguments from all functions.
From \re{eqn:J_K1}, the optimal transmission power is $T^{\star}_K(B)=B$, which is increasing in $B$.
We now apply Lemma~\ref{lem:monotonic} to establish that Theorem~\ref{thm:monotone} hold for $k< K$.
Let $F(B,T) = {\mi}(T) + \bar{J}_{k+1}(B-T)$, according to \re{eqn:J_k}.
Let $T_{l}(B)=0$, $T_{u}(B)=B,$ which are non-decreasing in $B$.
To apply Lemma~\ref{lem:monotonic}, it is sufficient to show that each term in $F$ has non-decreasing differences in $(B,T)$.
Since ${\mi}(T)$ is independent of $B$, trivially ${\mi}(T)$ has non-decreasing differences in $(B,T)$.
To show that $g(B-T)\triangleq \bar{J}_{k+1}(B-T)$ has non-decreasing differences in $(B,T)$, we note that
$g(y+\delta)-g(y)\leq g(x+\delta)-g(x)$ for $x\leq y, \delta\geq 0,$ since $g(x)=\bar{J}_{k+1}(x)$ is concave in $x$ from Theorem~\ref{thm:concave}.
Substituting $x=B-T', y=B-T, \delta=B'-B$, we then obtain  \re{eqn:incdiff} with $F(B,T)=g(B-T)$.
From Theorem~\ref{thm:concave}, the objective function in \re{eqn:J} is concave, thus $T^{\star}(B)$ is unique.
From Lemma~\ref{lem:monotonic}, $T^{\star}(B)$  is thus non-decreasing in $B$, $k\in\mathcal{K}.$
%\end{proof}

\section{Proof of Corollary~\ref{cor:K=2}}
%\begin{proof}[Proof of Corollary~\ref{cor:K=2}]
%For convenience we omit the subscripts in $B_1, H_1, T_1$.
Since $K=2$ and full SI is available, from \re{eqn:gaussianmi}, \re{eqn:bellman}  we get
\be\label{eqn:cor:1}
J_1(\snr_1, B_1)=\max_{0\leq T\leq B_1} g(T),
\ee
where the objective function is given by
%\ben
$g(T) \triangleq \log_2(1+\snr_1 T )
%\\
%\log(1+\snr T)
 + \log_2(1+\snr_2 \min\{\Bmax,B_1-T+H_1\}). %, \; 0\leq T\leq B.
$
%\een

Suppose $H_1> \Bmax$. Then $\min\{\Bmax,B_1-T+H_1\}=\Bmax$ given $T\leq B_1$.
%$g(T)=\log(1+T\snr_1)+\log(1+\Bmax\snr_2)$ for $0\leq T\leq B$.
The optimal $T$ that solves \re{eqn:cor:1} is then %\re{eqn:cor:K=2a}.
\be\label{eqn:cor:K=2a}
T_1^{\star}= B_1 \mbox{ if } H> \Bmax.
\ee

Suppose $H_1\leq \Bmax$. Assume that $0\leq T \leq B_1+H_1-\Bmax.$ Then $\min\{\Bmax,B_1-T+H_1\}=\Bmax$, and so the optimal $T$
to maximize $g(T)$ subject to $0\leq T \leq B_1+H_1-\Bmax$ is given by the largest value in the variable space, namely $B_1+H_1-\Bmax.$
%\ee
Thus, in general the optimal solution in \re{eqn:cor:1} satisfies $T_1^{\star}\geq [B_1+H_1-\Bmax]^+$, where $[x]^+\triangleq \max(0,x)$.
Without loss of generality, we can thus express $T_1^{\star}$ as
\be\label{eqn:dom1}
T_1^{\star}=\arg\max_{[B_1+H_1-\Bmax]^+\leq T\leq B_1} g(T)
\ee
if $H_1\leq \Bmax$.
Now if $T  \geq [B_1+H_1-\Bmax]^+$, we have %$g(T)$ becomes
%\ben
$g(T) = \log_2(1+\snr_1 T )+ \log_2(1+\snr_2  (B_1-T+H_1))$,
%\een
%$\min\{\Bmax,B-T+H\}= B-T+H$ in $g(T)$ and so $g(T)$ becomes .
which is differentiable and concave. Observe that $\widetilde{T}$ in \re{eqn:tildeT} solves the equation $g'(\widetilde{T})=0$, i.e., $\widetilde{T}$ is the optimal solution for the {\em unconstrained} optimization problem $\max g(T)$.
By concavity of $g(T)$, we can then obtain \re{eqn:dom1} as
\ben\label{eqn:dom2}
&&T_1^{\star}
=
\arg\max_{[B_1+H_1-\Bmax]^+\leq T\leq B_1} g(T) \nonumber \\
&&=
\left\{
\begin{array}{ll}
B_1, & \widetilde{T}>B_1; \\
\widetilde{T}, & [B_1+H_1-\Bmax]^+\leq \widetilde{T}\leq B_1; \\
\hspace{0cm}
[B_1+H_1-\Bmax ]^{+}, & \widetilde{T}< [B_1+H_1-\Bmax]^+
\end{array}
\right .
%\;\;\;\;\;
\een
if $H_1\leq \Bmax$.
%That is, if $B+H-\Bmax\leq \widetilde{T}\leq B$, then $\widetilde{T}$ is also the optimal solution $T^{\star}$ for the {\em constrained} optimization problem \re{eqn:dom1}; otherwise, if $\widetilde{T}>B$ or $\widetilde{T}<B+H-\Bmax$, then the optimal solution is given by the boundary of the solution space.
By re-writing the above conditions in terms of $B_1$ and combining the result with \re{eqn:cor:K=2a}, we  then obtain $T_1^{\star}$ as stated in Corollary~\ref{cor:K=2}.
%\end{proof}

%----------------------------------------------------------------
%\bibliographystyle{IEEEtran}
%%\bibliographystyle{alpha}
%% argument is your BibTeX string definitions and bibliography database(s)
%\bibliography{IEEEabrv,./../relay,./../TUebibfile_all}


\end{document}